\documentclass[draftclsnofoot,onecolumn,12pt]{IEEEtran}
\usepackage{amsmath,amsfonts}
\usepackage{bm}
\usepackage{amssymb}
\usepackage{enumerate}
\usepackage{algorithmic}
\usepackage{array}
\usepackage{textcomp}
\usepackage{stfloats}
\usepackage{url}
\usepackage{verbatim}
\usepackage{graphicx}
\usepackage{cite}
\usepackage{hyperref}
\usepackage{amsmath}
\usepackage{algorithm}
\allowdisplaybreaks[4]
\usepackage{float} 
\usepackage{multirow} 
\usepackage{ntheorem}
\usepackage{balance}
\usepackage{enumitem}
\usepackage{caption}
\usepackage{subcaption}
\usepackage{arydshln}
\usepackage{upgreek}
\usepackage{setspace}
\usepackage{epsf}
\usepackage{upgreek}
\usepackage{balance}
\usepackage{footnote}
\usepackage{color}

\usepackage{setspace}
\theorembodyfont{\upshape}
\theoremheaderfont{\it}
\qedsymbol{\square}
\theoremseparator{:}

\newtheorem{lemma}{\indent Lemma}
\newtheorem*{proof}{\indent Proof}
\newtheorem{remark}{\indent Remark}

\newtheorem{proposition}{\indent Proposition}
\makeatletter

\newcommand{\Rmnum}[1]{\expandafter\@slowromancap\romannumeral #1@}
\hypersetup{
    colorlinks=true,
    linkcolor=blue,
    filecolor=blue,
    urlcolor=blue,
    citecolor=blue,
}
\begin{document}
\begin{sloppypar}
\bstctlcite{BSTControl}
\title{Energy-Efficient Hybrid Beamfocusing for Near-Field Integrated Sensing and Communication}
\author{\mbox{Wenhao Hu, \IEEEmembership{Student Member, IEEE}}, \mbox{Zhenyao He, \IEEEmembership{Student Member, IEEE}}, \mbox{Wei Xu, \IEEEmembership{Fellow, IEEE}}, \mbox{Yongming Huang, \IEEEmembership{Fellow, IEEE}}, \mbox{Derrick Wing Kwan Ng, \IEEEmembership{Fellow, IEEE}}, and \mbox{Naofal Al-Dhahir, \IEEEmembership{Fellow, IEEE}}
}

\markboth{ 
        }
{Energy-Efficient Hybrid Beamfocusing for Near-Field Integrated Sensing and Communication}
\maketitle
\begin{abstract}

Integrated sensing and communication (ISAC) is a pivotal component of sixth-generation (6G) wireless networks, leveraging high-frequency bands and massive multiple-input multiple-output (M-MIMO) to deliver both high-capacity communication and high-precision sensing. However, these technological advancements lead to significant near-field effects, while the implementation of M-MIMO \mbox{is associated with considerable} hardware costs and escalated power consumption. In this context, hybrid architecture designs emerge as both hardware-efficient and energy-efficient solutions. Motivated by these considerations, we investigate the design of energy-efficient hybrid beamfocusing for near-field ISAC under two distinct target scenarios, i.e., a point target and an extended target. Specifically, we first derive the closed-form Cram\'{e}r-Rao bound (CRB) of joint angle-and-distance estimation for the point target and the Bayesian CRB (BCRB) of the target response matrix for the extended target. Building on these derived results, we minimize the CRB/BCRB by optimizing the transmit beamfocusing, while ensuring the energy efficiency (EE) of the system and the quality-of-service (QoS) for communication users. To address the resulting \mbox{nonconvex problems}, we first utilize a penalty-based successive convex approximation technique with a fully-digital beamformer to obtain a suboptimal solution. Then, we propose an efficient alternating \mbox{optimization} algorithm to design the analog-and-digital beamformer. \mbox{Simulation} results indicate that joint distance-and-angle estimation is feasible in the near-field region. However, the adopted hybrid architectures inevitably degrade the accuracy of distance estimation, compared with their fully-digital counterparts. Furthermore, enhancements in system EE would compromise the accuracy of target estimation, unveiling a nontrivial tradeoff. 


\end{abstract}
\begin{IEEEkeywords}
Integrated sensing and communication, near-field, beamfocusing, Cram\'{e}r-Rao bound, Bayesian CRB.
\end{IEEEkeywords}
\vspace{-1em}
\section{Introduction}

\IEEEPARstart{R}{eal}-time location awareness and ubiquitous transmission are emerging as defining features of future sixth-generation (6G) wireless networks \cite{FLiuIntegrated}. These advancements are crucial in enhancing connectivity and facilitating situational awareness across a wide range of applications, such as smart cities and industrial automation. In this context, integrated sensing and communication (ISAC) has been recognized as a promising technology for 6G. Specifically, ISAC enables effective spectrum sharing and hardware reuse for both communication and sensing, thereby enhancing spectral efficiency and reducing hardware costs, which supports various emerging services including environmental monitoring and autonomous driving \cite{ZhenyaoHeUnlocking},\cite{shi2023intelligent}. Therefore, both industry and academia are vigorously pursuing the development of ISAC to create a unified signal processing framework and develop innovative hardware platforms \cite{JAndrewZhangEnabling}.\par

One of the main challenges in establishing a unified framework for ISAC lies in the design of practical waveforms \cite{WXuEdge}. Typically, ISAC waveform design is categorized into the time-frequency domain and spatial domain. On the one hand, in the time-frequency domain, waveform research has primarily focused on technological innovations in modulation and signal embedding. For example, the authors of \cite{GEAFrankenDoppler} utilized orthogonal frequency division multiplexing (OFDM) signals for target velocity estimation. Also in \cite{YLiuDesign}, researchers investigated radar pulse signal transmission in a multiple-input multiple-output (MIMO) OFDM system, enhancing the communication rate by transmitting multiple OFDM symbols in a single pulse. Building upon these developments, \cite{THuangMAJoRCom} introduced an inter-pulse modulation design, where a radar pulse is encoded as a single communication codeword. \par
On the other hand, in the spatial domain, waveform research has mainly concentrated on beamforming design for resource allocation in both sensing and communication. For instance, in \cite{FLiuMUMIMO}, the authors aimed to minimize the mean squared error (MSE) between the designed beampattern and the reference beampattern for sensing. Similarly, \cite{XLiuJoint} minimized a radar loss function, which is comprised of a weighted combination of beampattern MSE and mean squared cross-correlation. Moving beyond directly optimizing beampatterns, other studies have focused on investigating the nontrivial tradeoff between communication and sensing performance metrics. The authors of \cite{BKChalisePerformance} optimized beamforming to effectively balance the tradeoff between target detection probability and communication achievable rate. Furthermore, in a full-duplex ISAC scenario, \cite{ZHeFullDuplex} tackled the problem of maximizing system sum rate and minimizing power consumption, subject to radar signal-to-interference-plus-noise ratio (SINR) constraints. Moreover, the studies in \cite{LiuZiangJoint} designed uplink and downlink beamforming vectors via maximizing the weighted sum rate, while simultaneously enhancing the beamforming gain towards the target for sensing. Lastly, \cite{FLiuCramerRao} investigated the beamforming design for a point target and an extended target, respectively, with the objective of minimizing the Cram\'{e}r-Rao bound (CRB) while satisfying SINR requirement constraints for communication users. These studies are devoted to enhancing estimation accuracy and communication rates. However, the limitations imposed by the number of antennas restrict the spatial multiplexing capability and spatial resolution of the array, thereby impacting the overall communication and sensing performance of the system. \par
In practice, to provide additional degrees of freedom (DoF) and enhanced spatial resolution, massive multiple-input multiple-output (M-MIMO) technology has emerged as a viable solution \cite{WXuToward}. However, this technology significantly increases both hardware costs and power consumption of systems, hindering their widespread implementations. To strike an effective balance between performance and hardware cost, hybrid architectures are indispensable for the practical realization of M-MIMO \cite{OEAyachSpatially}. In this context, \cite{FLiuHybrid} and \cite{XWangPartially} studied ISAC hybrid beamforming designs from different perspectives. In particular, M-MIMO's high power demands not only increased operational costs but also raised environmental concerns, rendering energy efficiency (EE) a critical design metric for M-MIMO \cite{KNRSVPrasadEnergy}. This focus on EE is essential for realizing the objectives of green ISAC and ensuring sustainable development \cite{YaoJiachengEnergy}. For instance, \cite{AlluRavitejaRobust} explored the design of a robust ISAC system aimed at maximizing system EE through resource optimization. Similarly, \cite{ZHeEnergy} optimized beamforming designs to improve ISAC EE in fully-digital architectures, while \cite{LYouBeam} addressed the same issue within the framework of hybrid architectures. In addition, \cite{ZouJiaqiEnergy} introduced a novel performance metric, namely sensing EE, and proposed a joint optimization framework that simultaneously considers both sensing EE and communication EE.\par

As exploration of M-MIMO technology progresses, fundamental changes occur in electromagnetic propagation characteristics, notably transitioning from the far-field region to the near-field region. This shift introduces a new resource dimension related to communication distance, allowing for the concentration of signal energy within specific distances and angles toward a target or user, known as beamfocusing \cite{CuiMingyaoNearField}. Indeed, the availability of the distance dimension offers novel approaches to distance estimation in sensing \cite{YDHuangNearField}. For near-field ISAC, \cite{ZWangNearField} proposed the beamfocusing design for a point target considering the distance estimation, but it did not fully consider multiple types of targets, limiting its applicability in more complex sensing scenarios. Moreover, \cite{DGalappaththigeNear} optimized the beamfocusing design under a fully-digital architecture. However, since M-MIMO typically employs hybrid architectures in practical applications, the proposed designs face limitations in real-world deployments. Although our previous work \cite{WenhaoHuIntegrated} investigated the beamfocusing design for an extended target under partially connected architectures, it neither demonstrated the availability of the near-field distance dimension, nor developed beamfocusing designs for point target scenarios. Therefore, the ISAC beamfocusing designs for M-MIMO with hybrid architectures necessitate further investigation, especially for estimating different target characteristics. \par
In practice, the joint waveform design for ISAC in near-field scenarios introduces several technical challenges: \textit{i)} The introduction of the distance dimension complicates the modeling of performance metrics for targets and users; \textit{ii)} While hybrid architecture-based M-MIMO can reduce hardware costs, the design of beamfocusing becomes more complex. Moreover, designing  the hybrid beamfocusing scheme to achieve high-precision distance and angle estimation while ensuring the EE of the system poses a significant challenge; \textit{iii)} Due to the limited DoFs in hybrid architectures, critical challenges arise, such as the potential nonexistence of the CRB for target response matrix (TRM) estimation in the case of an extended target. Motivated by these discussions, we investigate the M-MIMO beamfocusing design for hybrid analog-and-digital architectures. The main contributions of this paper are \mbox{as follows}:
\begin{itemize}[label=\textbullet]
  \item We establish a rectangular coordinate system to characterize the near-field channel for ISAC and derive the theoretical CRB for angle-and-distance estimation of a point target under the hybrid architecture. For an extended target, due to the limitations of DoFs and the absence of a conventional CRB for TRM estimation under the hybrid architecture, we resort to exploiting the Bayesian CRB (BCRB) to estimate the TRM.
  
  \item We formulate optimization problems aimed at minimizing the CRB of angle-and-distance estimation for a point target and minimizing the BCRB of TRM estimation for an extended target while ensuring the system EE and communication QoS requirements. To address the inherent challenges of nonconvex design problems, we develop a two-stage algorithm. Specifically, we first introduce an equivalent fully-digital beamfocusing approach to simplify the original problem and leverage the successive convex approximation (SCA) technique to address the nonconvexities of the problems, yielding a suboptimal solution. Then, we propose an efficient block coordinate descent algorithm to approximate the solution via a more practical hybrid beamfocusing design. 
 
  \item Numerical results demonstrate the effectiveness of the proposed algorithms. The near-field beamfocusing design enables simultaneous estimation of both angle and distance, demonstrating superior performance compared to the two-dimensional multiple signal classification ($\mathrm{2}\mathrm{D}$ MUSIC) algorithm. Additionally, our results reveal a nontrivial tradeoff between target estimation performance and system EE. In particular, as the EE of the system improves, the accuracy of the target estimation decreases.
\end{itemize}
This paper is organized as follows. \mbox{In Section \Rmnum{2}}, the system setting, near-field channel model, performance metrics, and problem formulation are described. \mbox{Section \Rmnum{3}} analyzes the CRB for angle-and-distance estimation of a point target and presents the corresponding beamfocusing solution. \mbox{In Section \Rmnum{4}}, the BCRB for TRM estimation of an extended target is derived and the beamfocusing solution is obtained. \mbox {Section \Rmnum{5}} discusses the iterative optimization of hybrid analog-and-digital beamformer, implemented in both fully-connected and partially-connected architectures. \mbox {In Section \Rmnum{6}}, the simulation results are presented. \mbox {Section \Rmnum{7}} summarizes the paper.

\textit{Notations:} Bold uppercase and lowercase letters stand for the matrices and column vectors, respectively. $\mathbb{C}^{M \!\times\! N}$ signifies $M\! \times \!N$ dimensional complex space. $ \mathbb{N}\in\{0,1,\dots\} $ denotes the set of nonnegative integers. The trace, expectation, and Kronecker product operators are presented as $ \mathrm{Tr}\{\!\cdot\!\}$, $ \mathbb{E}\{\!\cdot\!\} $, and $\otimes$, respectively. Real and imaginary parts of $x$ are represented by $\mathrm{Re}(x)$ and $\mathrm{Im}(x)$, respectively. Operations $ |x| $ and $\lceil x \rceil$ indicate the absolute value and ceiling value of $ x $, respectively. Nuclear norm, Frobenius norm, and $l_2$-norm are denoted by $ \|\!\cdot\!\|_{*} $, $ \|\!\cdot\!\|_\mathrm{F} $, and $\|\!\cdot\!\|_{2}$, respectively. Transpose, conjugate, and Hermitian conjugate are written as $(\cdot)^T$, $(\cdot)^{*}$, and $(\cdot)^H$, respectively. The $ i $-th row and the $ j $-th column element of matrix $ \mathbf{A} $ is denoted by $ [\mathbf{A}]_{i,j} $, an $N \!\times\! N$ dimensional identity matrix is indicated as $\mathbf{I}_{N}$. The rank of a matrix is represented as $ \mathrm{rank}(\cdot) $. $ \lambda_\mathrm{max}(\mathbf{A})$ denotes the maximum eigenvalue of the matrix $ \mathbf{A} $. The operation $ \mathrm{arg}(\cdot) $ represents the computation of the phase angle of a vector or matrix. The operation $ \mathrm{vec}(\cdot) $ denotes the matrix vectorization operation. The operation $\frac{\partial\mathbf{A}}{\partial\bm\gamma_i} $ denotes the partial derivative of the matrix $\mathbf{A}$ with respect to the $ i$-th element of vector $ \bm\gamma$. $\mathcal{CN}(\bm\mu,\,\bm\Sigma)$ is the distribution of a circularly symmetric complex Gaussian random vector with mean vector $\bm\mu$ and covariance matrix $ \bm\Sigma $.

\section{System Model and Problem Formulation}
\vspace{0.6em}

\subsection{System Model}
We consider a near-field hybrid multi-user multiple-input single-output (MU-MISO) ISAC system, where the BS is equipped with $ N_\mathrm{t} $ transmit antennas and $ N_\mathrm{t}^\mathrm{RF} $ RF chains, cf., Fig. \ref{figNearField}. Meanwhile, the BS serves $ K $ single-antenna users while simultaneously detecting the existence of a single target, satisfying $ K \!\leq\! N_\mathrm{t}^\mathrm{RF} \!\leq\! N_\mathrm{t} $. To improve the performance of sensing echo signals, we assume that a fully-digital antenna array with $ N_\mathrm{r} $ antennas is employed at the BS's receiver side, positioned separately from the transmitting array to eliminate transceiver self-interference \cite{XWangPartially}. Let $ \mathbf{s}_k\!\in\!\mathbb{C}^{L \times 1}, \,\forall k\!\in\!\{1,\dots,K\}$, denote the downlink baseband data symbol vector for the $k$-th user, where $ L $ represents the length of a communication/radar frame. Then, define $ \mathbf{S}\!=\![\mathbf{s}_{1},\dots,\mathbf{s}_{K}]^T\!\in\!\mathbb{C}^{K \times L} $ as the downlink data symbol matrix. Without loss of generality, we assume that data streams for users are statistically independent and the covariance matrix of $\mathbf{S}$ is given by $\mathbf{R_S}\!=\!\frac{1}{L}\mathbf{SS}^H\!\approx\!\mathbf{I}_{K}$. This approximation property is primarily justified by the law of large numbers. When the frame length $ L $ is sufficiently large, the matrix $ \mathbf{R_S} $ converges to its expected value and approximation property remains asymptotically \cite{LuShihangAjoint}. Under the hybrid beamfocusing architecture, the transmit signal  $ \mathbf{X} \in \mathbb{C}^{N_\mathrm{t} \times L} $ over duration $L$ from the ISAC BS can be written as
\vspace{-0.1em}
\begin{equation} \label{transmitSignal}
  \mathbf X=\mathbf{T}_\mathrm{A}\mathbf{T}_\mathrm{D} \mathbf{S},
  \vspace{-0.1em}
\end{equation}
where $ \mathbf{T}_\mathrm{D}\!\triangleq\![\mathbf{t}_{\mathrm{D},1},\mathbf{t}_{\mathrm{D},2},\dots,\mathbf{t}_{\mathrm{D},K}]\!\in\!\mathbb{C}^{N_\mathrm{t}^\mathrm{RF}\times K} $ denotes the digital beamformer, $ \mathbf{t}_{\mathrm{D},k}\in\mathbb{C}^{N_\mathrm{t}^\mathrm{RF}\times 1} $ is the digital beamfocusing vector for the $k$-th user, and $\mathbf{T}_\mathrm{A}\!\in\!\mathbb{C}^{N_\mathrm{t}\times N_\mathrm{t}^\mathrm{RF}}$ denotes the analog beamformer. 
\begin{figure}[t]
  \centering
  \includegraphics[scale=0.5]{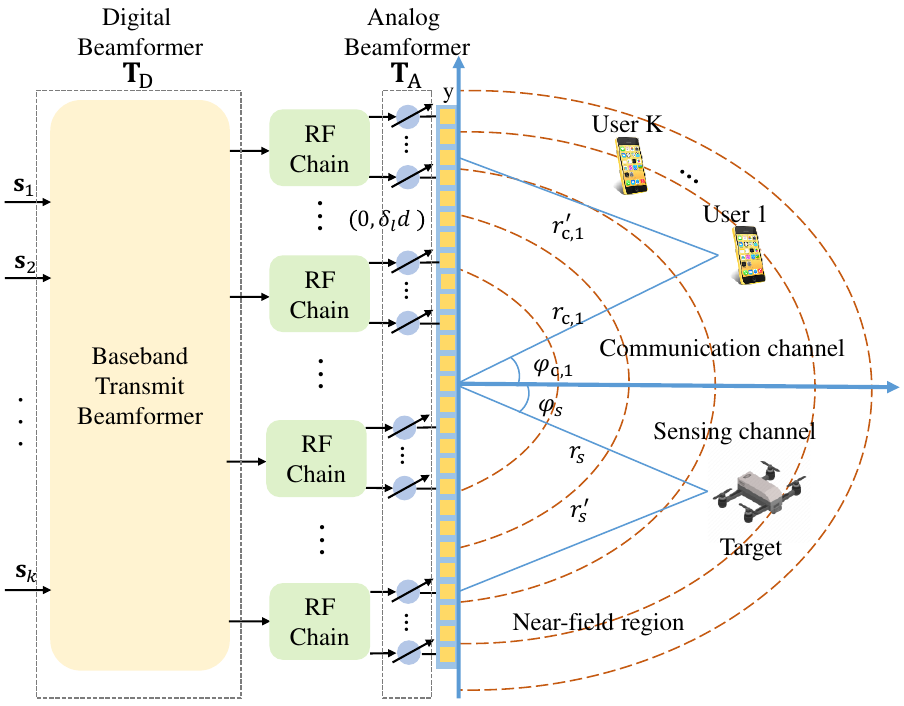}
  \caption{A near-field ISAC system with hybrid beamfocusing.}
  \label{figNearField}
  \vspace{-0.6em}
\end{figure}
The sample covariance matrix of $ \mathbf{X} $ is written as 
\vspace{-0.2em}
\begin{equation} \label{covMatrix}
    \mathbf{R}_\mathrm{\mathbf{X}} =\frac{1}{L}\mathbf{X}\mathbf{X}^{H}=\mathbf{T}_\mathrm{A}\mathbf{T}_\mathrm{D}\mathbf{T}_\mathrm{D}^{H}\mathbf{T}_\mathrm{A}^{H},
    \vspace{-0.2em}
\end{equation}
where $ \mathbf{R}_\mathrm{\mathbf{X}}\in\mathbb{C}^{N_\mathrm{t}\times N_\mathrm{t}} $. Note that due to the limitation of the number of RF chains, the maximum rank of matrix $\mathbf{R}_\mathrm{\mathbf{X}} $ is upper bounded by $ N_\mathrm{t}^\mathrm{RF} $, i.e., $ \mathrm{Rank}(\mathbf{R}_\mathrm{\mathbf{X}}) \leq N_\mathrm{t}^\mathrm{RF}\! <\! N_{\mathrm{t}}$.
\vspace{-0.2em}
\setlength{\skip\footins}{8pt}  

\subsection{Near-Field Channel Model}
For simplicity, we consider the utilization of a uniform linear array (ULA) at the transmitter to construct the near-field channel model, primarily to simplify the theoretical analysis and provide a foundational framework for geometric configurations\footnote{ Considering practical applications, our future studies will consider the near-field channel model based on the uniform planar array (UPA) to enhance the adaptability and accuracy of the system \cite{shi2024secrecy}.}, as commonly adopted in the literature, e.g., \cite{AlluRavitejaRobust, ZHeEnergy, JZouEnergy, LYouBeam}. As shown in Fig. \ref{figNearField}, a rectangular coordinate system is established at the center of ULA, with the array element spacing, $ d $, set to half of the signal carrier wavelength $ \lambda $, i.e., $ d\!=\!\frac{\lambda}{2} $ \cite{MCuiNear}. The coordinate of the $ l $-th element of ULA is expressed as $ (0,\delta_l d) $, with $ \delta_l \! =\! \frac{2l-N_\mathrm{t}-1}{2},\,\forall l\!\in\!\{1,2,\cdots,N_\mathrm{t}\} $. Then, the distance between a specific user equipment (UE)/target and the $ l $-th element of the ULA is $ r_l'(\varphi)\!=\!\sqrt{r^2+\delta_l^2 d^2-2r\delta_ld\sin\varphi} $, where $ r $ represents the distance between the UE/target and the center of the transmitting ULA and $ \varphi \!\in\! (-\pi/2,\pi/2) $ denotes the specific angle of the UE/target. Then, the near-field steering vector is given by
\begin{equation} \label{steeringVector}
 \mathbf{b}(r,\varphi)=[e^{-j\frac{2\pi}{\lambda}(r_1^{\prime}(\varphi)-r)},\dots,e^{-j\frac{2\pi}{\lambda}(r_{N_\mathrm{t}}^{\prime}(\varphi)-r)}]^T.
\end{equation}
Applying the Taylor series expansion $ \sqrt{1+x}\approx 1+\frac{1}{2}x-\frac{1}{8}x^2 + \mathit{o}(x^2) $ to $ r_l'(\varphi)\!=\!r\sqrt{1+\frac{1}{r^2}(\delta_l^2 d^2-2r\delta_ld\sin\varphi)}$, the difference $ r_l'(\varphi)-r $ is approximated as
\begin{equation} \label{disApproximate}
  r_l^{\prime}(\varphi)-r\approx-\delta_{l}d\sin\varphi+\delta_{l}^2d^2\frac{\cos^2\varphi}{2r}.
\end{equation}
Then, we consider the spherical wave channel model \cite{CuiMingyaoNearField} for the high-frequency near-field scenarios, including both line-of-sight (LoS) and non-line-of-sight (NLoS) paths. The channel vector $ \mathbf{h}_k\in\mathbb{C}^{N_\mathrm{t}\times 1} $, between the BS and the $k$-th user is given by
\vspace{-0.3em}
 \begin{equation} \label{channelModel}
  \mathbf{h}_k\!=\!\beta_{L}\mathbf{b}(r_{\mathrm{c},k},\varphi_{\mathrm{c},k})\!+\!\sum\limits_{i=1}^{I_k}\! \beta_{i} \mathbf{b}(r_{\mathrm{c},k,i},\varphi_{\mathrm{c},k,i}),
  \vspace{-0.3em}
\end{equation}
\noindent where $ \beta_{L}\!\in\!\mathbb{C} $ denotes the LoS complex gain, $ r_{\mathrm{c},k} $ and $ \varphi_{\mathrm{c},k} $ represent the distance and angle from the BS to the $ k $-th user, respectively, $ I_k $ denotes the number of NLoS paths for the $ k $-th user, $ \beta_{i}\!\in\!\mathbb{C} $ denotes the complex gain of the $ i $-th NLoS path, and $ r_{\mathrm{c},k,i} $ and $ \varphi_{\mathrm{c},k,i} $ represent the distance and angle from the BS to the $ i $-th scatter of the $ k $-th user, respectively. Additionally, since the near-field propagation typically occurs under high-frequency conditions, not only conventional path loss but also molecular absorption loss should be taken into account. Accordingly, the gains for $ \beta_{L}$ and $ \beta_{i}$ are modeled based on the high-frequency propagation characteristics in \cite{LinAdaptive} and \cite{CLinSubarrayBased}.

\subsection{Sensing Performance Metric}
The BS transmits $\mathbf{X}$ to support simultaneously user communication and single-target sensing. The received echo signal at the BS receiver, denoted by $ \mathbf{Y} \in \mathbb{C}^{N_\mathrm{r}\times L} $, is written as 
\vspace{-0.2em}
\begin{equation} \label{echoSignal}
  \mathbf{Y}=\mathbf{B}\mathbf{X} + \mathbf{N},
  \vspace{-0.2em}
\end{equation}
\noindent where $ \mathbf{B}\!\in\!\mathbb{C}^{N_\mathrm{r}\times N_\mathrm{t}} $ denotes the TRM \cite{BTangSpectrally}. In line with different target models, the TRM may admit different forms. Typically, the TRM is constructed as one of the two models below.

\subsubsection{Point Target}
In this context, the target is treated as an unstructured point located away from the BS \cite{MBellInformation}. The TRM is modeled as
\vspace{-0.4em}
\begin{equation} \label{pointTarRepMat}
    \mathbf{B}=\mu\mathbf{b}_\mathrm{r}(r_{\mathrm{s}},\varphi_{\mathrm{s}})\mathbf{b}_\mathrm{t}^H(r_{\mathrm{s}},\varphi_{\mathrm{s}}),
    \vspace{-0.4em}
\end{equation}
\noindent where $ \mu\in\mathbb{C} $ represents the complex-valued reflection coefficient of the target, $\mathbf{b}_\mathrm{r}(r_{\mathrm{s}},\varphi_{\mathrm{s}})\in\mathbb{C}^{N_\mathrm{r}\times 1}$ and $\mathbf{b}_\mathrm{t}(r_{\mathrm{s}},\varphi_{\mathrm{s}})\in\mathbb{C}^{N_\mathrm{t}\times 1}$ denote receive and transmit steering vectors, respectively, which are defined following \eqref{steeringVector}. Here, $ r_{\mathrm{s}} $ and $ \varphi_{\mathrm{s}}$ represent the corresponding distance and angle for the target, respectively. We consider a colocated MIMO radar architecture, where the direction of departure and direction of arrival to the same target are identical \cite{XLiuJoint}.
\subsubsection{Extended Target}
When the target is in close proximity to the BS, it is typically modeled as an extended target. In such a case, the target should be modeled as a surface comprising numerous distributed points-like scatterers \cite{ALeshemInformation}. Usually, the TRM is modeled as 
\vspace{-0.2em}
\begin{equation} \label{extendedTarRepMat}
  \mathbf{B}=\sum^{N}_{i=1}\mu_i\mathbf{b}_\mathrm{r}(r_{{\mathrm{s}},i},\varphi_{{\mathrm{s}},i})\mathbf{b}_\mathrm{t}^H(r_{{\mathrm{s}},i},\varphi_{{\mathrm{s}},i}),
  \vspace{-0.4em}
\end{equation}
\noindent where $ N $ is the number of scatterers, $\mu_i\!\in\!\mathbb{C}$ denotes the reflection coefficients of the $ i $-th scatterer, and $ \varphi_{{\mathrm{s}},i} $ and $ r_{{\mathrm{s}},i} $ represent the angle and distance of the $ i $-th scatterer, respectively. Notably, we assume that all point-like scatterers are located in the same range bin to facilitate the \mbox{target estimation \cite{FLiuCramerRao}}. \par
To improve the target estimation performance of the near-field ISAC system, the CRB and BCRB are employed as the sensing performance metrics, which are subsequently minimized through the optimization problem formulation. For a point target scenario, we employ the CRB for assessing the estimation performance, serving as a lower bound for the error variance of any unbiased estimator \cite{SMKayFundamentals}. In contrast to far-field ISAC scenarios, the CRB for the estimation of parameters is composed of both the target angle $ \varphi $ and the distance $ r $ in the near-field. Detailed procedures for analyzing the CRB will be presented later in Section \Rmnum{3}. 
On the other hand, for an extended target scenario, the BS lacks knowledge regarding the quantity of scatterers, as well as the positions of the scatterers. As a consequence, we directly estimate the TRM as an entity, as opposed to estimating the angles and distances of all scatterers. However, due to the fact that the number of RF chains in hybrid architectures should be much smaller than the number of antennas, a rank deficiency in the sample covariance matrix occurs \cite{JXuSecure}, rendering the nonexistence of the CRB for TRM estimation. This challenge motivates us to exploit the BCRB to estimate the TRM, with detailed derivations provided in Section \Rmnum{4}. On the other hand, to ensure the quality of service for communication users, we adopt SINR as the performance metric for communication component. 

\vspace{-1em}
\subsection{Communication Performance Metric}
The communication signal vector received at the $ k $-th user over duration $ L $ is expressed as
\vspace{-0.4em}
\begin{equation}\label{receSig}
  \mathbf{y}_k^T=\mathbf{h}_k^H\mathbf{t}_{k}\mathbf{s}_k^T+\!\sum_{i=1,i\neq k}^K\!\mathbf{h}_k^H\mathbf{t}_{i}\mathbf{s}_i^T+\mathbf{n}_k^T,
  \vspace{-0.6em}
\end{equation}
\noindent where $ \mathbf{n}_k\in\mathbb{C}^{L\times 1}$ denotes the additive white Gaussian noise (AWGN) at the $k$-th user, i.e., $\mathbf{n}_k\!\sim\!\mathcal{CN}(\mathbf{0},\,\sigma_{k}^{2}\mathbf{I}_L) $, and $ \sigma_k^2 $ is the noise power, $ \mathbf{t}_k\triangleq\mathbf{T}_{\mathrm{A}}\mathbf{t}_{\mathrm{D},k} $ represents the beamfocusing vector for the $ k $-th user. Furthermore, we assume that a channel estimation stage is performed prior to beamfocusing design, allowing the BS and users to perfectly acquire the related channel state information (CSI) \cite{CGTsinosJoint, WXuDisentangled}. Then, the SINR of the $k$-th user over duration $ L $ is given by 
\vspace{-0.2em}
\begin{equation}\label{SINR}
  \Gamma_{k}=\frac{|\mathbf{h}_k^H\mathbf{t}_{k}|^2}{\sum\limits_{i\neq k}^K|\mathbf{h}_k^H\mathbf{t}_{i}|^2+\sigma^2_k}.
  \vspace{-0.2em}
\end{equation}
The total communication achievable rate of the system is then expressed as $ R \triangleq \sum_{k=1}^K\log_2(1+\Gamma_k) $. Due to the use of high-frequency band and M-MIMO arrays, the system has inevitably exhibited substantial power consumption \cite{KNRSVPrasadEnergy}. Indeed, the pursuit of reducing transmission power while ensuring high-speed transmission emerges as a key design direction, drawing attention to the design of energy-efficient transmission. The system EE is defined as the ratio of the achievable rate to total power consumption, expressed as
\vspace{-0.4em}
\begin{equation}\label{EE}
  \eta_\mathrm{EE}=\frac{R}{P_\mathrm{total}},
  \vspace{-0.4em}
\end{equation}
where we adopt a linear power model to assess the total power consumption of the system \cite{ZHeEnergy}, denoted by $P_\mathrm{total}$, and \mbox{expressed as}
\vspace{-1.2em}
\begin{equation}\label{totalPower}
  P_{\mathrm{total}}=\frac{1}{\rho}\sum_{k=1}^K\|\mathbf{t}_{k}\|^2+P_0, 
  \vspace{-0.4em}
\end{equation}
where $ \rho \in (0,1]$ denotes the power amplifier coefficient, $ P_0 $ indicates the static circuit power consumption of the BS, and $ \sum_{k=1}^K\|\mathbf{t}_{k}\|^2 $ is the total radiated signal power.
\vspace{-0.4em}
\subsection{Problem Formulation}

The digital beamformer $ \mathbf{T}_{\mathrm{D}} $ and analog beamformer $ \mathbf{T}_{\mathrm{A}} $ are optimized to minimize the CRB of angle-and-distance estimation for a point target and the BCRB of TRM estimation for an extended target, while guaranteeing both system EE and communication QoS for users. Specifically, the problem is formulated as
\vspace{-0.7em}
\begin{subequations}\label{CRBopti} 
  \begin{align}
    \mathop{\mathrm{minimize}}\limits_{\mathbf{T}_\mathrm{D},\,\mathbf{T}_\mathrm{A}}  \qquad &\mathrm{Tr}\{\mathrm{CRB}\}\, \,\mathrm{or} \,\,\mathrm{Tr}\{\mathrm{BCRB}\}  \tag{\ref{CRBopti}}\\
    \mathrm{subject \;to} \qquad & \|\mathbf{T}_\mathrm{A}\mathbf{T}_\mathrm{D}\|^2_\mathrm{F}\leq P, \label{CRBopti_powCon} \\
      \qquad &   \Gamma_{k}\geq\Gamma_\mathrm{th},\,\forall k\in\{1,\cdots,K\},\label{CRBopti_SINRCon}\\ 
      \qquad &   \eta_\mathrm{EE}\geq \eta_\mathrm{th},\label{CRBopti_eeCon}\\
      \qquad &   \mathbf{T}_\mathrm{A} \in \mathcal{A},\label{CRBopti_unitCon}
   \end{align}
\end{subequations}
where \eqref{CRBopti_powCon} represents the total transmit power constraint, constraint \eqref{CRBopti_SINRCon} ensures a lower bound on the minimum required SINR for each user to guarantee the communication QoS, $ \eta_\mathrm{th} $ in constraint \eqref{CRBopti_eeCon} serves as a threshold to provide a minimum system EE requirement, \eqref{CRBopti_unitCon} represents the constant modulus constraint for the analog beamformer, which takes the following form according to the fully-connected or partially-connected architecture \cite{AAroraHybrid}. Specifically, every antenna in a fully-connected architecture is linked to all RF chains, while in a partially-connected architecture, the antennas are divided into $ N_\mathrm{t}^\mathrm{RF} $ different groups. Therefore, constraint \eqref{CRBopti_unitCon} indicates that the non-zero elements of the analog beamformer $ \mathbf{T}_\mathrm{A} $ belong to the set $\mathcal{A}\!=\!\{\mathcal{A}_\mathrm{F}, \mathcal{A}_\mathrm{P}\}$, given by
\vspace{-0.4em}
\begin{align}\label{unitConFulPart}
  &\mathcal{A}_{\mathrm{F}}=\left\{\mathbf{T}_\mathrm{A}\Big|\left|\left[\mathbf{T}_\mathrm{A}\right]_{p,q}\right|=1,\forall p,q\right\},\\
  &\mathcal{A}_{\mathrm{P}}=\left\{\mathbf{T}_\mathrm{A}\Big|\left|\left[\mathbf{T}_\mathrm{A}\right]_{p,q}\right|=1,\forall p,\forall q=\left\lceil\frac{p}{I_{\mathrm{g}}}\right\rceil\right\},
\end{align}
where $ I_{\mathrm{g}}\!\triangleq\!N_\mathrm{t}/N_\mathrm{t}^\mathrm{RF} $ denotes the number of antennas connected to the same RF chain within each group. The problem in \eqref{CRBopti} involves an objective function of an intractable CRB matrix and several constraints, which are nonconvex with respect to the two tightly coupled variables, i.e., digital and analog beamformers. To address these issues, we initially take into account the digital and analog beamformers as an equivalent fully-digital beamformer, as well as temporarily drop the unrelated constraints. The performance of this setup serves as an upper-bound benchmark for other hybrid architecture designs. Subsequently, we propose an algorithm to tackle the equivalent fully-digital near-field CRB optimization problem. Eventually, an alternating optimization framework for hybrid beamformers is exploited to approximate an equivalent fully-digital beamformer.

\begin{figure*}[b]  
  \hrulefill  
  \begin{equation}\label{expressionMatrixA}
    \mathbf{A}=\frac{2|\mu|^2L}{\sigma_{n}^2}
    \begin{bmatrix}
     \mathrm{Tr}(\dot{\mathbf{B}}_{r}^H\dot{\mathbf{B}_{r}}\mathbf{R}_\mathrm{\mathbf{X}})-\frac{|\mathrm{Tr}(\dot{\mathbf{B}}_{r}^H\mathbf{B}\mathbf{R}_{\mathrm{X}})|^2}{\mathrm{Tr}(\mathbf{B}^H\mathbf{B}\mathbf{R}_{\mathrm{X}})}& \mathrm{Tr}(\dot{\mathbf{B}}_{r}^H\dot{\mathbf{B}}_{\varphi}\mathbf{R}_\mathrm{\mathbf{X}})-\frac{\mathrm{Re}\{\mathrm{Tr}(\dot{\mathbf{B}}_{r}^H\mathbf{B}\mathbf{R}_{\mathrm{X}})\mathrm{Tr}(\dot{\mathbf{B}}_{\varphi}^H\mathbf{B}\mathbf{R}_{\mathrm{X}})^*\}}{\mathrm{Tr}(\mathbf{B}^H\mathbf{B}\mathbf{R}_{\mathrm{X}})}\\
      \mathrm{Tr}(\dot{\mathbf{B}}_{\varphi}^H\dot{\mathbf{B}}_r\mathbf{R}_{\mathrm{X}})-\frac{\mathrm{Re}\{\mathrm{Tr}(\dot{\mathbf{B}}_{r}^H\mathbf{B}\mathbf{R}_{\mathrm{X}})\mathrm{Tr}(\dot{\mathbf{B}}_{\varphi}^H\mathbf{B}\mathbf{R}_{\mathrm{X}})^*\}}{\mathrm{Tr}(\mathbf{B}^H\mathbf{B}\mathbf{R}_{\mathrm{X}})} & \mathrm{Tr}(\dot{\mathbf{B}}_{\varphi}^H\dot{\mathbf{B}_{\varphi}}\mathbf{R}_\mathrm{\mathbf{X}})-\frac{|\mathrm{Tr}(\dot{\mathbf{B}}_{\varphi}^H\mathbf{B}\mathbf{R}_{\mathrm{X}})|^2}{\mathrm{Tr}(\mathbf{B}^H\mathbf{B}\mathbf{R}_{\mathrm{X}})}\\
    \end{bmatrix}
\end{equation}
\end{figure*}
\vspace{0.4em}
\section{Beamfocusing Design for Point Target}
In contrast to the CRB of angle estimation in the far-field region, in this section, we begin by deriving the CRB of angle-and-distance estimation for a near-field point target. Then, we concentrate on the optimization problem that minimizes the sum of CRB for the point target's distance and angle estimation. Finally, we propose an SCA-based iterative algorithm to acquire an effective solution.


\vspace{-0.6em}
\subsection{CRB Analysis for Near-Field Point Target}
We first focus on deriving the CRB of angle-and-distance estimation for the scenarios considering a point target. The CRB expression is obtained by constructing the Fisher information matrix (FIM) with respect to the unknown parameters \cite{SMKayFundamentals}. Let the unknown parameters vector be defined as $ \bm{\delta}\!=\![\bm{\phi}^T \tilde{\bm\mu}^T]^T $, where $ \bm\phi\!=\![r,\varphi]^T $, $ \tilde{\bm\mu}\!=\!\left[\mu_{\mathrm{R}},\mu_{\mathrm{I}}\right]^T $, and $\mu_{\mathrm{R}}$ and $\mu_{\mathrm{I}}$ denote the real and imaginary parts of $ \mu $, respectively. Simultaneously, by vectorizing the received echo signal matrix in \eqref{echoSignal}, we have
\vspace{-0.4em}
\begin{equation} \label{vecReceSignal}
  \mathbf{y}=\mathrm{vec}(\mathbf{Y})=\bm\upsilon+\mathbf{n},
  \vspace{-0.4em}
\end{equation}
where $\bm\upsilon=\mathrm{vec}(\mu\mathbf{B}\mathbf{X})$, $\mathbf n=\mathrm{vec}(\mathbf N)$, and $ \mathbf{n} $ is a complex Gaussian noise vector, each element of which has zero mean and variance $ \sigma_{n}^2 $. Then, we have $ \mathbf{y}\sim\mathcal{CN}(\bm\upsilon,\,\mathbf{R}_{n}) $ with $ \mathbf{R}_{n}=\sigma_{n}^{2}\mathbf{I}_{N_\mathrm{r}L} $. Furthermore, the FIM under the complex Gaussian model is expressed as \cite{SMKayFundamentals}
\vspace{-0.4em}
\begin{equation}\label{FIMDefinition}
  [\mathbf{J}]_{i,j}\!=\!\mathrm{Tr}\!\left( \!\mathbf{C}^{-1}\frac{\partial\mathbf{C}}{\partial\bm\xi_i}\mathbf{C}^{-1}\!\frac{\partial\mathbf{C}}{\partial\bm\xi_j}\!\right)\!+\!2\mathrm{Re}\!\left\{\!\frac{\partial\bm\omega^H}{\partial\bm\xi_i}\mathbf{C}^{-1}\!\frac{\partial\bm\omega}{\partial\bm\xi_j}\!\right\},
  \vspace{-0.4em}
\end{equation}
where $ i,j\!\in\! \{1,\dots,n_{p}\}$, $ n_{{p}}$ denotes the dimension of $ \bm\xi $, $ \bm\xi $ denotes the unknown parameters vector, and $ \bm\omega$ and $ \mathbf{C} $ denote the mean vector and covariance matrix of $ \bm\xi $, respectively. According to \eqref{FIMDefinition}, the FIM $ \mathbf{J}_{\bm\delta}$ for the unknown parameters vector, $ \bm\delta $, is given by
\vspace{-0.4em}
\begin{equation}\label{FIMDiviDedMatrixPoint}
  \setlength{\arraycolsep}{3.5pt} 
  \mathbf{J}_{\bm\delta}\!=\!\left[\!\begin{array}{cc:cc}
  \!J_{{rr}}\!&\!J_{{r}\varphi}\!&\!J_{{r}\mu_{\mathrm{R}}}\!&\!J_{{r}\mu_{\mathrm{I}}}\!\\
  \!J_{\varphi {r}}\!&\!J_{\varphi\varphi}\!&\!J_{\varphi\mu_{\mathrm{R}}}\!&\!J_{\varphi\mu_{\mathrm{I}}}\!\\
  \hdashline
  \!J_{{r}\mu_{\mathrm{R}}}\!&\!J_{\varphi \mu_{\mathrm{R}}}\!&J_{\mu_{\mathrm{R}}\mu_{\mathrm{R}}}\!&0\!\\
  \!J_{{r}\mu_{\mathrm{I}}}\!&\!J_{\varphi\mu_{\mathrm{I}}}\!&\!0\!&J_{\mu_{\mathrm{I}}\mu_{\mathrm{I}}}\!\\
  \end{array}\!\right]\!=\!\left[\!\begin{array}{c:c}
    \!\mathbf{J}_{\bm\phi\bm\phi} \!&\! \mathbf{J}_{\bm\phi\tilde{\bm\mu}}\!\\
    \hdashline
    \!\mathbf{J}_{\bm\phi\tilde{\bm\mu}}^T\! &\! \mathbf{J}_{\tilde{\bm\mu}\tilde{\bm\mu}}\! \\
    \end{array}\!\right]\!,
    \vspace{-0.4em}
\end{equation} 
where the FIM is partitioned into four block matrices corresponding to the vectors $ \bm{\phi} $ and $ \tilde{\bm{\mu}} $. The expressions of block matrices $\mathbf{J}_{\bm\phi\bm\phi}$, $\mathbf{J}_{\bm\phi\tilde{\bm\mu}}$, and $ \mathbf{J}_{\tilde{\bm\mu}\tilde{\bm\mu}}$ are shown below, respectively, and derived in Appendix \ref{FIMPointDeriveAppendix}:
\vspace{-0.3em}
\begin{align}\label{FIMElements}
 &J_{uv}=\frac{2|\mu|^2L}{\sigma_{n}^2}\mathrm{Tr}(\dot{\mathbf{B}}_u^H\dot{\mathbf{B}}_v\mathbf{R}_\mathrm{\mathbf{X}}),\\
  &\mathbf{J}_{\bm\phi\tilde{\bm\mu}}=\frac{2L}{\sigma_{n}^2}\mathrm{Re}\left\{
    \begin{bmatrix}
      \mu^*\mathrm{Tr}(\dot{\mathbf{B}}_{r}^H\mathbf{B}\mathbf{R}_\mathrm{\mathbf{X}})\\
      \mu^*\mathrm{Tr}(\dot{\mathbf{B}}_{\varphi}^H\mathbf{B}\mathbf{R}_\mathrm{\mathbf{X}})
    \end{bmatrix} 
   [1,j] \right\},\\
  &\mathbf{J}_{\tilde{\bm\mu}\tilde{\bm\mu}}=\frac{2L}{\sigma_{n}^2}\mathbf{I}_2\mathrm{Tr}(\mathbf{B}^H\mathbf{B}\mathbf{R}_\mathrm{\mathbf{X}}), 
\end{align}
where $J_{uv},\forall u,v\in\bm\phi $, denotes the element of matrix $\mathbf{J}_{\bm\phi\bm\phi}$, and $ \dot{\mathbf{B}}_{r}=\frac{\partial\mathbf{B}}{\partial r} $, $ \dot{\mathbf{B}}_{\mathrm{\varphi}}=\frac{\partial\mathbf{B}}{\partial\varphi} $. According to the Schur complement property, the inverse of the FIM is represented using block matrices, shown as
\vspace{-0.2em}
\begin{equation}\label{FIMDiviDedMatrixPoint}
  \setlength{\arraycolsep}{2.5pt} 
  \mathbf{J}_{\bm\delta}^{-1}\!=\!\left[\!\begin{array}{c:c}
    \!\mathbf{A}^{-1}& \!-\mathbf{A}^{-1}\mathbf{J}_{\bm\phi\tilde{\bm\mu}}\mathbf{J}_{\tilde{\bm\mu}\tilde{\bm\mu}}\\
  \hdashline
  \!-\mathbf{J}_{\tilde{\bm\mu}\tilde{\bm\mu}}\mathbf{J}_{\bm\phi\tilde{\bm\mu}}^T\mathbf{A}^{-1}&  \mathbf{J}_{\tilde{\bm\mu}\tilde{\bm\mu}}^{-1}\!+\!\mathbf{J}_{\tilde{\bm\mu}\tilde{\bm\mu}}^{-1}\mathbf{J}_{\bm\phi\tilde{\bm\mu}}^T\mathbf{A}^{-1}\mathbf{J}_{\bm\phi\tilde{\bm\mu}}\mathbf{J}_{\tilde{\bm\mu}\tilde{\bm\mu}}^{-1}\\
  \end{array}\!\right]\!,
  \vspace{-0.2em}
\end{equation} 
where $\mathbf{A}\!\triangleq\!\mathbf{J}_{\bm\phi\bm\phi}\!-\!\mathbf{J}_{\bm\phi\tilde{\bm\mu}}\mathbf{J}_{\tilde{\bm\mu}\tilde{\bm\mu}}^{-1}\mathbf{J}_{\bm\phi\tilde{\bm\mu}}^T$. Then, the CRB with respect to the unknown parameters vector $ \bm{\phi} $ is expressed as
\vspace{-0.2em}
\begin{equation}\label{CRBexpressionTrace}
  \mathrm{CRB}(\bm{\phi})\!=\!\mathbf{A}^{-1}\!=\!\!\left(\!\mathbf{J}_{\bm\phi\bm\phi}\!-\!\mathbf{J}_{\bm\phi\tilde{\bm\mu}}\mathbf{J}_{\tilde{\bm\mu}\tilde{\bm\mu}}^{-1}\mathbf{J}_{\bm\phi\tilde{\bm\mu}}^T\!\right)\!^{-1}\!,
\end{equation}
where the specific expression of matrix $ \mathbf{A} $ is shown in \eqref{expressionMatrixA}. Furthermore, the CRB can be attained by exploiting maximum likelihood estimation (MLE) \cite{FLiuCramerRao}.

\vspace{-0.6em}
\subsection{Problem Formulation}\label{pointTargetDesignSection}
Based on the above analysis, the optimization problem of minimizing the CRB for the point target's distance and angle estimation can be formulated as follow
\begin{subequations}\label{CRBoptiExpression} 
   \begin{align}
    \color{blue}
    {\mathop{\mathrm{minimize}}\limits_{\mathbf{T}_\mathrm{D},\,\mathbf{T}_\mathrm{A}}}\quad &  {\mathrm{Tr}\{\mathrm{CRB}(\bm{\phi})\}}  \tag{\ref{CRBoptiExpression}}\\
    {\mathrm{subject \;to}} \quad & {\eqref{CRBopti_powCon},\,\eqref{CRBopti_SINRCon},\,\eqref{CRBopti_eeCon},\,\eqref{CRBopti_unitCon}.}
  \end{align}
\end{subequations}
Due to the complex form of the objective function, we equivalently transform problem \eqref{CRBoptiExpression} into a more tractable one as shown in the following proposition.
\begin{proposition}\label{proSchurConv}
  Minimizing $\mathrm{CRB}(\bm{\phi})$ is equivalent to solving the following problem
\end{proposition}
\vspace{-0.4em}
\begin{subequations}\label{CRBconstraint} 
  \begin{align}
    \mathop{\mathrm{minimize}}\limits_{\mathbf{W},\,\bm\Xi}\qquad   &\mathrm{Tr}(\bm\Xi ^{-1})   \tag{\ref{CRBconstraint}}\\
    \mathrm{subject \;to} \qquad  &\begin{bmatrix} \label{CRBconstraint_shurCon}
                                      \mathbf{J}_{\bm\phi\bm\phi}-\bm\Xi & \mathbf{J}_{\bm\phi\tilde{\bm\mu}}\\
                                      \mathbf{J}_{\bm\phi\tilde{\bm\mu}}^T & \mathbf{J}_{\tilde{\bm\mu}\tilde{\bm\mu}} 
                                        \end{bmatrix}\succeq  \mathbf{0},\\
                                      &\bm\Xi\succeq \mathbf{0},\label{CRBconstraint_AuxSemi}\\
                                      &\eqref{CRBopti_powCon},\,\eqref{CRBopti_SINRCon},\,\eqref{CRBopti_eeCon},
  \end{align}
\end{subequations}
where $ \bm\Xi $ is an auxiliary optimization variable matrix, and \mbox{$\mathbf{W}\!\triangleq\!\mathbf{T}_{\mathrm{A}}\mathbf{T}_{\mathrm{D}} $}.
\begin{proof}
  Please refer to Appendix \ref{propositionProof} for the proof.
\end{proof} 
\par
In the above proposition, we have introduced $ \mathbf{W} $ as the equivalent fully-digital beamformer. To proceed, we define $\mathbf{W}_k\!=\!\mathbf{w}_k\mathbf{w}_k^H\!\succeq\!\mathbf{0}$ with $ \mathrm{rank}(\mathbf{W}_k)\!=\!1, \forall k $, where $ \mathbf{w}_{k} $ is the $ k $-th column of $ \mathbf{W} $ and signifies the beamfocusing vector for the $ k $-th user. The problem \eqref{CRBconstraint} is recast as
\vspace{-0.2em}
\begin{subequations}\label{CRBequ} 
  \begin{align}
    \mathop{\mathrm{minimize}}\limits_{\{\mathbf{W}_k\}_{k=1}^K,\bm\Xi}\quad   &\mathrm{Tr}(\bm\Xi ^{-1})   \tag{\ref{CRBequ}}\\
    \mathrm{subject \;to} \quad  &\begin{bmatrix} \label{CRBequ_shurCon}
                          \mathbf{J}_{\bm\phi\bm\phi}-\bm\Xi & \mathbf{J}_{\bm\phi\tilde{\bm\mu}}\\
                          \mathbf{J}_{\bm\phi\tilde{\bm\mu}}^T & \mathbf{J}_{\tilde{\bm\mu}\tilde     {\bm\mu}}
                          \end{bmatrix}\succeq \mathbf{0},\\ 
                        &\bm\Xi\succeq\mathbf{0},\label{CRBequ_AuxSemi}\\
                        &\sum\limits_{k=1}^K\mathrm{Tr}(\mathbf{W}_k)\leq P,\label{CRBequ_powCon}\\
                        &\mathrm{Tr}(\mathbf{H}_k\!\mathbf{W}_k\!)\!\geq\!\Gamma_\mathrm{th}\!\left(\!\sum\limits_{i \neq k}^K\!\mathrm{Tr}(\mathbf{H}_k\mathbf{W}_i\!)\!+\!\sigma_k^2\!\right)\!,\label{CRBequ_SINRCon}\\
                        &\sum\limits_{k=1}^K\!\log_2(1\!+\!\Gamma_k^{\prime})\!-\!\frac{\eta_\mathrm{th}}{\rho}\!\sum\limits_{k=1}^K\!\mathrm{Tr}(\mathbf{W}_{k})\!-\!P_0\eta_\mathrm{th}\!\geq\! 0\label{CRBequ_EErateCon},\\
                        &\mathbf{W}_1,\dots,\mathbf{W}_k\succeq \mathbf{0}, \label{CRBequ_semi}\\
                        &\mathrm{rank}(\mathbf{W}_k)\leq1,\quad k=1,\dots,K, \label{CRBequ_rankCon}
  \end{align}
\end{subequations}
where $\mathbf{H}_k \!=\! \mathbf{h}_k\mathbf{h}_k^H$ and $ \Gamma_k'\!=\!\frac{\mathrm{Tr}(\mathbf{H}_k\mathbf{W}_k)}{\sum_{i\neq k}^K\mathrm{Tr}(\mathbf{H}_k\mathbf{W}_i)+\sigma^2_k}$. The primary advantage of \eqref{CRBequ} lies in transforming the original complex optimization problem in \eqref{CRBexpressionTrace} into a positive semidefinite optimization problem \cite{ZqLuoSemidefinite}. Now, we are ready to develop an effective approach to address problem \eqref{CRBequ} in the \mbox{next subsection.}

\vspace{-0.6em}
\subsection{Proposed Algorithm for Solving Problem \eqref{CRBequ}}\label{pointTargetProblemFormulation}
It is evident that constraints \eqref{CRBequ_EErateCon} and \eqref{CRBequ_rankCon} lead to the nonconvexity of problem \eqref{CRBequ}. To handle the nonconvex rank-one constraint in \eqref{CRBequ_rankCon}, the semidefinite relaxation (SDR) technique is typically leveraged \cite{ZqLuoSemidefinite}. However, ensuring that the obtained solution satisfies the rank-one constraint is generally challenging. When a high-rank solution is obtained due to constraint relaxation, various techniques can be employed to approximate a low-rank solution, such as eigenvalue approximation or Gaussian randomization \cite{ZqLuoSemidefinite}. These methods, however, come with different drawbacks. On the one hand, the eigenvalue approximation method may result in a decrease in solution accuracy, particularly when the objective function involves inversion operations. On the other hand, Gaussian randomization, especially for complex optimization problems, can produce significantly high computational complexity. As a remedy, we first equivalently convert the nonconvex rank-one constraint into a new expression that facilitates algorithm development in the sequel, as illustrated in the lemma below.
\vspace{-0.4em}
\begin{lemma}
  For any positive semidefinite matrix $ \mathbf{X} \!\in\! \mathbb{C}^{M_{\mathrm{t}} \times M_{\mathrm{t}}} $, the rank-one constraint for $\mathbf{X} $ is equivalent to
\end{lemma} 
\vspace{-0.4em}
\begin{equation}\label{equivalentRankCon}
  \mathrm{rank}(\mathbf{X})\leq1\Leftrightarrow\left\|{\mathbf{X}}\right\|_*-\left\|{\mathbf{X}}\right\|_2\leq0.
\end{equation}
\vspace{-2em}
\begin{proof}
 $ \left\|\mathbf{X}\right\|_*=\sum_{i=1}^{M_{\mathrm{t}}}\sigma_i$, $\left\|\mathbf{X}\right\|_{2}=\mathop{\max}\limits_{i}\{\sigma_{i}\} $, where $ M_{\mathrm{t}} $ is the total number of singular values, and $ \sigma_i $ is the $ i $-th singular value of $ \mathbf{X} $. Then, the inequality $ \left\|\mathbf{X}\right\|_* \geq \left\|\mathbf{X}\right\|_{2} $ holds. Thus, equality holds if and only if $ \mathbf{X} $ is of rank one.
 \vspace{-0.4em}
\end{proof}
\par
Based on the above lemma, although we can transform the rank-one constraint in \eqref{CRBequ_rankCon} into its equivalent form, it remains nonconvex. To overcome this difficulty, we tackle the problem by resorting to a penalty-based method \cite{XYuRobust}. Based on Proposition \ref{proposition2}, the equivalent constraint is augmented as a penalty term to the objective function, transforming the problem into
\vspace{-0.4em}
\begin{subequations}\label{CRBPointSPenalty} 
  \begin{align}
    \mathop{\mathrm{minimize}}\limits_{\{\mathbf{W}_k\}_{k=1}^K,\bm\Xi}\quad   &\mathrm{Tr}(\bm\Xi ^{-1})+\frac{1}{\gamma}\sum^K\limits_{k=1}\left(\left\|{\mathbf{W}_k}\right\|_*-\left\|{\mathbf{W}_k}\right\|_2\right)   \tag{\ref{CRBPointSPenalty}}\\
    \mathrm{subject \;to} \quad  &\eqref{CRBequ_shurCon}-\eqref{CRBequ_semi},
  \end{align}
\end{subequations}
\noindent where $ \gamma \ll 1 $ serves as a penalty factor penalizing the violation of the equivalent constraint.  
\vspace{-0.2em}
\begin{proposition}\label{proposition2}
  For the penalty factor $ \gamma $, let $ \{\mathbf{W}_k^{\prime}\}_{k=1}^K $ denote the optimal solutions of problem \eqref{CRBPointSPenalty}. As $ \gamma $ gets sufficiently small, i.e. $ \gamma \rightarrow 0 $, the optimal solution to problem \eqref{CRBequ} is attained at the limit point of sequence $\left\{\{\mathbf{W}_k^{(n)}\}_{k=1}^K\right\}_{n\in \mathbb{N}}$.
  \vspace{-0.2em}
\end{proposition}
\begin{proof}
  See \cite[Appendix C]{XYuRobust}.
  \vspace{-0.6em}
\end{proof}
\par
Therefore, \eqref{CRBequ} and \eqref{CRBPointSPenalty} are equivalent, i.e., share an identical optimal solution. Indeed, the objective function of \eqref{CRBPointSPenalty} takes the form of differences of convex functions. Hence, we employ the iterative SCA algorithm to solve it. In particular, we aim to establish a lower bound for $ \left\|\mathbf{W}_k\right\|_2 $ by employing its first-order Taylor series expansion as a surrogate at each iteration. Specifically, the lower bound expression for $ \left\|\mathbf{W}_k\right\|_2 $ at the $ n $-th iteration of the SCA is expressed as
\vspace{-0.1em}
\begin{align}
    \left\|\mathbf{W}_k\right\|_2 & \!\geq\!\left\|\mathbf{W}_k^{(n-1)}\right\|_2\!+\!\mathrm{Tr}\!\left\{\!\mathbf{u}_{\max}\!\left(\mathbf{W}_k^{(n-1)}\right)\!\mathbf{u}_{\max}^{H}\!\left(\mathbf{W}_k^{(n-1)}\right)\right.\notag\\
    &\left.\times\left(\mathbf{W}_k-\mathbf{W}_k^{(n-1)}\right)\right\} = \widetilde{\mathbf{W}}_k^{(n)},
\end{align}
\noindent where $ \mathbf{u}_{\max}(\mathbf{W}_k^{(n-1)}) $ denotes the eigenvector corresponding to the largest singular value of the matrix $ \mathbf{W}_k $ at the $ (n-1) $-th iteration. Concurrently, we denote the lower bound expression as $ \widetilde{\mathbf{W}}_k^{(n)} $. Subsequently, problem \eqref{CRBPointSPenalty} at the $ n $-th iteration of SCA is reformulated as
\vspace{-0.3em}
\begin{subequations}\label{CRBPointSPenaltySCA} 
  \begin{align}
    \mathop{\mathrm{minimize}}\limits_{\{\mathbf{W}_k\}_{k=1}^K,\bm\Xi}\quad   &\mathrm{Tr}(\bm\Xi ^{-1})\!+\!\frac{1}{\gamma}\!\sum^K\limits_{k=1}\!\bigg(\!\left\|{\mathbf{W}_k}\right\|_*\!-\! \widetilde{\mathbf{W}}_k^{(n)}\!\bigg)   \tag{\ref{CRBPointSPenaltySCA}}\\
    \mathrm{subject \;to} \quad   &\eqref{CRBequ_shurCon}-\eqref{CRBequ_semi}.
  \end{align}
\end{subequations}
Moreover, the concavity of constraint on \eqref{CRBequ_EErateCon} primarily arises from the $\log_2(\cdot)$ function. Thus, we focus on the left-hand side of \eqref{CRBequ_EErateCon}, and express $ \log_2(1+\Gamma_k^{\prime})$ as
\vspace{-0.4em}
\begin{equation}\label{rateEqual}
  \log_2\!\left(\!\sum\limits_{i=1}^K\!\mathrm{Tr}(\mathbf{H}_k\!\mathbf{W}_i\!)\!+\!\sigma_k^2\!\right)\!-\!\log_2\!\left(\!\sum\limits_{i\neq k}^K\!\mathrm{Tr}(\mathbf{H}_k\!\mathbf{W}_i\!)\!+\!\sigma_k^2\!\right)\!,
  \vspace{-0.4em}
\end{equation}
where the function $\log_2(\cdot)$ is a concave increasing function with respect to variables $ \{\mathbf{W}_k\}_{k=1}^K $, but the second term in \eqref{rateEqual} renders this constraint nonconvex. Therefore, we utilize the first-order Taylor expansion for the second term with respect to $ \mathbf{W}_i $, which establishes a lower bound of $\log_2(1+\Gamma_k')$ in the $ n $-th iteration of SCA as follows:
\vspace{-0.4em}
\begin{align}\label{rateLowerBound}
  &\log_2(1+\Gamma_k')\ge\log_2\left(\sum\limits_{i=1}^{K}\mathrm{Tr}(\mathbf{H}_k\mathbf{W}_i)+\sigma_k^2\right)\notag\\
  &\!-\!\left\{\!w_k^{(n\!-\!1)}\!+\!\frac{\log_2e}{2^{w_k^{(n\!-\!1)}}}\!\sum_{i\neq k}^{K}\!\mathrm{Tr}\!\left(\!\mathbf{H}_k(\mathbf{W}_i\!-\!\mathbf{W}_i^{(\!n-1\!)})\!\right)\!\right\}\!=\!R_k^{\prime},
\end{align}
\vspace{-0.4em}
\par
\noindent where $w_k^{(\!n-1\!)}\!=\!\log_{2}(\sum_{i\neq k}^{K}\!\mathrm{Tr}(\mathbf{H}_k\mathbf{W}_i^{(\!n-1\!)})\!+\!\sigma_{k}^{2}),\,\{\mathbf{W}_i^{\!(n-1\!)}\}\!_{i=1}^K $ denote the solutions obtained in the $ (n\!-\!1) $-th iteration. Then, \eqref{rateLowerBound} is adopted as a surrogate constraint in \eqref{CRBequ_EErateCon}. The problem \eqref{CRBPointSPenaltySCA} at the $ n $-th iteration of SCA is reformulated as
\vspace{-0.4em} 
\begin{subequations}\label{CRBPointSPenaltySCARk} 
  \begin{align}
    \mathop{\mathrm{minimize}}\limits_{\{\mathbf{W}_k\}_{k=1}^K,\bm\Xi}\quad   &\mathrm{Tr}(\bm\Xi ^{-1})\!+\!\frac{1}{\gamma}\!\sum^K\limits_{k=1}\!\bigg(\!\left\|{\mathbf{W}_k}\right\|_*\!-\! \widetilde{\mathbf{W}}_k^{(n)}\!\bigg)   \tag{\ref{CRBPointSPenaltySCARk}}\\
    \mathrm{subject \;to} \quad  & \sum\limits_{k=1}^K\!R_k^{\prime}\!-\!\frac{\eta_\mathrm{th}}{\rho}\!\sum\limits_{k=1}^K\!\mathrm{Tr}(\mathbf{W}_{k})\!-\!P_0\eta_\mathrm{th}\!\geq\! 0\label{CRBPointSPenaltySCARk_EErateCon},\\
    &\eqref{CRBequ_shurCon}-\eqref{CRBequ_SINRCon},\,\eqref{CRBequ_semi},
  \end{align}
\end{subequations}

\noindent which is jointly convex with respect to $\{\mathbf{W}_k\}_{k=1}^{K}$ and $\bm\Xi$. Consequently, standard convex program solvers such as CVX can be applied to acquire an optimal solution for problem \eqref{CRBPointSPenaltySCARk}. The complete procedure of the developed penalty-based SCA algorithm is outlined in \textbf{Algorithm \ref{algorithmPenaltySCAPoint}}.\par
\par
\begin{algorithm}[t]
  \caption{Penalty-Based SCA Algorithm for Problem \eqref{CRBPointSPenaltySCARk}}
  \label{algorithmPenaltySCAPoint}
  \begin{algorithmic}[1]
    \renewcommand{\algorithmicrequire}{\textbf{Input:}}
    \renewcommand{\algorithmicensure}{\textbf{Output:}}
    \vspace{-0.2em}
  \STATE \textbf{Initialize} a feasible point $ \{\mathbf{W}^{(0)}_k\}^{K}_{k=1},\,\bm\Xi^{(0)}$. Set the iteration index $ n=0 $, the number of maximum iterations $ N_\mathrm{max} $, the convergence tolerance $ 0 < \varepsilon \ll 1 $, and the penalty factor $ 0 < \gamma \ll 1 $. Calculate the initial objective value $ \mathrm{Obj}^{(0)} $.
  \REPEAT
    \STATE Update $ n=n+1 $;
    \STATE Solve the problem \eqref{CRBPointSPenaltySCARk} with $ \{\mathbf{W}^{(n-1)}_k\}^{K}_{k=1}$, $\bm\Xi^{(n-1)}$;
    \STATE Output the optimal solutions $ \{\mathbf{W}^{(n)}_k\}^{K}_{k=1}$, $\bm\Xi^{(n)}$ and the objective value $ \mathrm{Obj}^{(n)} $; 
  \UNTIL convergence. 
  \vspace{-0.2em}
  \end{algorithmic}
\end{algorithm}
\vspace{-0.4em}

\subsection{Computational Complexity and Convergence Analysis}\label{pointTargetComplexityConvergence}
Essentially, the computational complexity of Algorithm \ref{algorithmPenaltySCAPoint} hinges on solving problem \eqref{CRBPointSPenaltySCA}. Since the problem is a semidefinite programming (SDP) problem, it can typically be solved by exploiting the interior point method. Consequently, the computational complexity of the penalty-based SCA algorithm is $\mathcal{O}(\ln(\frac{1}{\epsilon})NN_\mathrm{t}^{6.5}{K}^{3.5}) $ \cite{KYWangOutage}, where $ \epsilon > 0 $ denotes the convergence tolerance of the interior point algorithm and $ N $ represents the required number of iterations for SCA. \par
Regarding the convergence of the iterative algorithm, the minimum value obtained from problem \eqref{CRBPointSPenaltySCA} serves as an upper bound of the original problem \eqref{CRBequ}, and the upper bound progressively tightens monotonically after $ n $ iterations of \eqref{CRBPointSPenaltySCA}. With reference to \cite{ABeckSequential}, optimal values achieved by the sequence $\left\{\{\mathbf{W}_k^{(n)}\}_{k=1}^K\right\}_{n\in \mathbb{N}}$ constitute a non-increasing sequence that converges to a Karush-Kuhn-Tucker (KKT) point of problem \eqref{CRBequ} in polynomial time.

\vspace{0.6em}
\section{Beamfocusing Design for Extended Target}
In this section, since the rank deficiency of the covariance matrix, the attainment of CRB for TRM estimation of an extended target in hybrid architectures is prevented, as proved in Appendix \ref{FIMExtendedDeriveAppendix}. As a compromise, we resort to the Bayesian technique to analyze the BCRB \cite{JTabrikianTheoretical} of TRM estimation for an extended target. Once TRM can be accurately estimated, the classical signal processing algorithms are applied to extract information from scatterers. If the number of scatterers $ N\leq\mathrm{rank}(\mathbf{B})$, the scatterers' paths are distinguishable, and the corresponding distances and angles can be extracted from the estimation of $ \mathbf{B} $ \cite{YDHuangNearField}. Conversely, if $ N\geq\mathrm{rank}(\mathbf{B}) $, differentiating all paths becomes impossible. In this case, an effective strategy is developed to identify and utilize the dominant paths for further information extraction.\par
\vspace{-0.4em}
\subsection{Bayesian CRB Analysis for Near-Field Extended Target}
Based on the above discussion, obtaining the CRB for TRM estimation for an extended target is very challenging, as shown in Appendix \ref{FIMExtendedDeriveAppendix}. The Bayesian estimation techniques are advantageous in the absence of unbiased minimum variance estimators. Therefore, we leverage Bayesian techniques to estimate matrix $ \mathbf{B} $. The formulation is
\vspace{-0.2em}
\begin{equation}\label{BCRBDefinition}
  \mathrm{BCRB}(\mathbf{B})=\mathbf{J}_{\bm{\delta}}^{\prime -1}=\left(\mathbf{J}_{1}+\mathbf{J}_{2}\right)^{-1},
  \vspace{-0.2em}
\end{equation}
where $\mathbf{J}_{1}$ and $\mathbf{J}_{2}$ form the FIM $\mathbf{J}_{\bm{\delta}}^{\prime}$, representing the impact of observed values and prior statistical information on the estimation error bound \cite{JTabrikianTheoretical}, respectively, which are given by
\begin{equation}\label{BCRBDDefPart_1}
  [\mathbf{J}_{1}]_{ij}\!=\!-\mathbb{E}\bigg\{\!\frac{\partial^{2}\ln p(\mathbf{y}|\bm{\delta})}{\partial\bm{\delta}_{i}\partial\bm{\delta}_{j}}\!\bigg\}, [\mathbf{J}_{2}]_{ij}\!=\!-\mathbb{E}\bigg\{\!\frac{\partial^{2}\ln p(\bm\delta)}{\partial\bm\delta_{i}\partial\bm\delta_{j}}\!\bigg\},
\end{equation}
where $p (\mathbf{y}|\bm{\delta})$ is the conditional probability density of the received data vector $ \mathbf{y} $ given the vector $ \bm{\delta} $ to be estimated, and $p(\bm{\delta})$ represents the prior probability distribution of the unknown parameters $ \bm{\delta} $.\par 
\vspace{-0.1em}
As demonstrated in Appendix \ref{FIMExtendedDeriveAppendix}, matrix $\mathbf{J}_{1}$ obtained under the conditional probability density $ p (\mathbf{y}|\bm{\delta}) $ is singular and equals to \eqref{FIMDiviDedMatrixExtendedCRLB}. Furthermore, by incorporating the statistical prior information of the parameters to be estimated, we can derive the BCRB of TRM estimation for an extended target. Based on the central limit theorem, as the number of distributed points increases in the same range bin, it is plausible to assume that statistical information for the response matrix of the extended target follows a circularly symmetric complex Gaussian (CSCG) distribution \cite{ZJWuExtendedTarget}. Without loss of generality, we define $ \bm\beta\!=\!\mathrm{vec}(\mathbf{B})\!\in\!\mathbb{C}^{N_\mathrm{t}N_\mathrm{r}} $ and the statistical information of $ \bm\beta $ is denoted as $ \bm\beta\!\sim\!\mathcal{CN}(\mathbf{0},\,\mathbf{R}_{\bm\beta\bm\beta}) $ with $ \mathbf{R}_{\bm\beta\bm\beta}\!=\!\sigma_{\bm\beta}^2\mathbf{I}_{N_\mathrm{t}N_\mathrm{r}} $. Subsequently, the statistical prior distribution of the vector $ \bm{\delta} $ is given by $ \bm{\delta}\!\sim\!\mathcal{CN}(\mathbf{0},\,\mathbf{R}_{\bm{\delta}}) $, where $ \mathbf{R}_{\bm{\delta}}\!=\!\frac{\sigma_{\bm\beta}^2}{2}\mathbf{I}_{2N_\mathrm{t}N_\mathrm{r}} $. The matrix $\mathbf{J}_{2}$, which is the inversion of prior covariance matrix of the unknown parameters vector $ \bm{\delta} $, is defined as $\mathbf{J}_{2}\!=\!\mathbf{R}_{\bm{\delta}}^{-1}\!=\!\frac{2}{\sigma_{\bm\beta}^2}\mathbf{I}_{2N_\mathrm{t}N_\mathrm{r}}$. The block matrices expression of Bayesian FIM $ \mathbf{J}_{\bm\delta}^{\prime}$ is formulated as 
\vspace{-0.4em}
\begin{equation}\label{BCRB}
  \mathbf{J}_{\bm\delta}^{\prime}=
\begin{bmatrix}
  \mathbf{J}_{\bm\beta_\mathrm{R}\bm\beta_\mathrm{R}}^{\prime}& \mathbf{J}_{\bm\beta_\mathrm{R}\bm\beta_\mathrm{I}}^{\prime}\\
  \mathbf{J}_{\bm\beta_\mathrm{I}\bm\beta_\mathrm{R}}^{\prime}&\mathbf{J}_{\bm\beta_\mathrm{I}\bm\beta_\mathrm{I}}^{\prime}
\end{bmatrix}.
\vspace{-0.2em}
\end{equation}
Then, the block matrices of FIM $\mathbf{J}_{\bm{\delta}}^{\prime}$ are expressed as
\vspace{-0.4em}
\begin{align}\label{FIMDiviDedMatrixExtendedElementBCRB}
  &\mathbf{J}_{\bm\beta_\mathrm{R}\bm\beta_\mathrm{R}}^{\prime}\!=\!\mathbf{J}_{\bm\beta_\mathrm{I}\bm\beta_\mathrm{I}}^{\prime}\!=\!\frac{2L}{\sigma_{n}^2}\!\mathrm{Re}\!\left\{\mathbf{R}_\mathrm{\mathbf{X}}^T\!\otimes\!\mathbf{I}_{N_\mathrm{r}}\!+\!\frac{\sigma_{n}^{2}}{\sigma_{\bm\beta}^2L}\mathbf{I}_{N_\mathrm{t}N_\mathrm{r}}\right\},\\
  &\mathbf{J}_{\bm\beta_\mathrm{I}\bm\beta_\mathrm{R}}^{\prime}\!=\!-\mathbf{J}_{\bm\beta_\mathrm{R}\bm\beta_\mathrm{I}}^{\prime}\!=\!\frac{2L}{\sigma_{n}^2}\!\mathrm{Im}\!\left\{\mathbf{R}_\mathrm{\mathbf{X}}^T\!\otimes\!\mathbf{I}_{N_\mathrm{r}}\!+\!\frac{\sigma_{n}^{2}}{\sigma_{\bm\beta}^2L}\mathbf{I}_{N_\mathrm{t}N_\mathrm{r}}\right\}.
\end{align}
Applying Lemma \ref{lemmaComplexInverse}, the trace of the BCRB for the TRM estimation of an extended target is given by
\begin{align}\label{BCRBpriorPart}
   \mathrm{Tr}\{\mathrm{BCRB(\mathbf{B})}\}&=\mathrm{Tr}(\mathbf{J}_{\bm{\delta}}^{\prime -1})\notag\\
   &=\frac{\sigma_{n}^2}{L}\mathrm{Tr}\left\{\left[\left(\mathbf{R}_\mathrm{\mathbf{X}}\!+\!\frac{\sigma_{n}^2}{\sigma_{\bm\beta}^2L}\mathbf{I}_{N_\mathrm{t}}\right)\!\otimes\!\mathbf{I}_{N_\mathrm{r}}\right]^{-1}\right\}\notag\\
   &=\frac{\sigma_{n}^2N_\mathrm{r}}{L}\mathrm{Tr}\left\{\left(\mathbf{R}_\mathrm{\mathbf{X}}\!+\!\frac{\sigma_{n}^2}{\sigma_{\bm\beta}^2L}\mathbf{I}_{N_\mathrm{t}}\right)^{-1}\right\}\!.
\end{align}
Compared to the CRB for a point target, the BCRB of TRM for an extended target can be expressed by a simpler closed-form expression in \eqref{BCRBpriorPart}. In fact, the BCRB of TRM leverages the distribution characteristics and prior information of the extended target, allowing the prior information to play a significant role in the estimation process. 
\begin{remark}
  \vspace{-0.2em}
  The BCRB is achievable by exploiting a linear minimum mean squared error (LMMSE) estimator.
  \vspace{-0.2em}
\end{remark} 
The LMMSE estimation vector $ \hat{\bm\beta} $ is expressed as
\vspace{-0.1em}
\begin{equation}
    \hat{\bm\beta}\!=\!\mathbf{R}_{\bm\beta\bm\beta}\!\left[L\!\left(\mathbf{R}^T_\mathrm{X}\!\otimes\!\mathbf{I}_{N_\mathrm{r}}\!\right)\!\mathbf{R}_{\bm\beta\bm\beta}\!+\!\sigma_{n}^2\mathbf{I}_{N_\mathrm{t}N_\mathrm{r}}\!\right]\!^{-1}\!\left(\mathbf{X}\!^T\!\otimes\!\mathbf{I}_{N_\mathrm{r}}\!\right)\!\mathbf{y}.
    \vspace{-0.1em}
\end{equation}
The reason behind this is that the LMMSE estimation of vector $ \mathbf{b} $ is a linear estimation problem in the presence of the prior knowledge that vector $ \mathbf{b} $ follows a CSCG distribution.


\vspace{-0.8em}
\subsection{Problem Formulation}
Following a point target optimization problem formulation in Sec. \ref{pointTargetDesignSection}, in an extended target scenario, the optimization objective is to minimize the BCRB while ensuring the power consumption and EE of the system, expressed as follows
\vspace{-0.2em}
\begin{subequations}\label{BCRBExtendeq} 
  \begin{align}
    \mathop{\mathrm{minimize}}\limits_{\{\mathbf{W}_k\}_{k=1}^K}\quad  &\mathrm{Tr}\{\mathrm{BCRB}(\mathbf{B})\}  \tag{\ref{BCRBExtendeq}}\\
    \mathrm{subject \;to} \quad    
                    &\sum\limits_{k=1}^K\mathrm{Tr}(\mathbf{W}_k)\leq P,\label{BCRBExtendeq_powCon}\\
                    &\mathrm{Tr}(\mathbf{H}_k\!\mathbf{W}_k\!)\!\geq\!\Gamma_\mathrm{th}\!\left(\!\sum\limits_{i \neq k}^K\!\mathrm{Tr}(\mathbf{H}_k\!\mathbf{W}_i\!)\!+\!\sigma_k^2\!\right)\!,\label{BCRBExtendeq_SINRCon}\\
                    &\sum\limits_{k=1}^K\!\log_2\!\left(1\!+\!\Gamma_k' \right)\!-\!\frac{\eta_\mathrm{th}}{\rho}\!\sum\limits_{k=1}^K\!\mathrm{Tr}(\mathbf{W}_{k})\!-\!P_0\eta_\mathrm{th}\!\geq\!0,\label{BCRBExtendeq_EErateCon}\\
                    &\mathbf{W}_1,\dots,\mathbf{W}_k\succeq \mathbf{0}, \label{BCRBExtendeq_semi}\\
                    &\mathrm{rank}(\mathbf{W}_k)\leq1,\quad k=1,\dots,K. \label{BCRBExtendeq_rankCon}
  \end{align}
\end{subequations}
Similar to Sec. \ref{pointTargetProblemFormulation}, the rank-one and \eqref{BCRBExtendeq_EErateCon} nonconvex constraints are also addressed by leveraging the penalty method and SCA techniques. Subsequently, the BCRB optimization problem at the $ n $-th iteration is formulated as 
\vspace{-0.2em}
\begin{subequations}\label{BCRBExtendPenaltySCA} 
  \begin{align}
    \mathop{\mathrm{minimize}}\limits_{\{\mathbf{W}_k\}_{k=1}^K}\quad  &\mathrm{Tr}\{\mathrm{BCRB}(\mathbf{B})\}\!+\!\frac{1}{\gamma}\!\sum^K\limits_{k=1}\!\bigg(\!\left\|{\mathbf{W}_k}\right\|_*\!-\! \widetilde{\mathbf{W}}_k^{(n)}\!\bigg)\!\tag{\ref{BCRBExtendPenaltySCA}}\\
    \mathrm{subject \;to} \quad  &\sum\limits_{k=1}^K R_k'\!-\!\frac{\eta_\mathrm{th}}{\rho}\!\sum\limits_{k=1}^K\!\mathrm{Tr}(\mathbf{W}_{k})\!-\!P_0\eta_\mathrm{th}\!\geq\!0,\\  
    &\eqref{BCRBExtendeq_powCon},\eqref{BCRBExtendeq_SINRCon},\eqref{BCRBExtendeq_semi}.
  \end{align}
\end{subequations}
The computational complexity and convergence analysis of Problem \eqref{BCRBExtendPenaltySCA} are similar to those in Sec. \ref{pointTargetComplexityConvergence}.\par

Based on the proposed optimization algorithm, the framework provides an effective approach for near-field target estimation. Building upon the structure of the parameter vector modeling for the target, the framework presents general applicability. Specifically, for near-field multi-target estimation, it enables the construction of a joint optimization problem by formulating FIM that incorporates the estimation parameters of multiple targets. In mixed near-field and far-field scenarios, the framework is applicable to different propagation models by formulating FIM of mixed target parameters. Moreover, in multi-static sensing systems, the framework incorporates FIMs of the estimation parameters at each sensing node to formulate a cooperative multi-static optimization problem. Therefore, the proposed framework can be effectively extended to multi-target estimation, mixed near-field and far-field targets estimation, and multi-static estimation scenarios.

\section{near-field Target Estimation Hybrid Precoding}
In this section, we aim to optimize both the analog and digital beamformers to minimize the Euclidean distance between the equivalent fully-digital beamformer and the product of analog and digital beamformers. The problem is formulated as follows
\vspace{-0.6em}
\begin{algorithm}
  \renewcommand{\algorithmicrequire}{\textbf{Input:}}
  \renewcommand{\algorithmicensure}{\textbf{Output:}}
  \caption{Partially-Connected Architecture for Hybrid Precoding}
  \label{partiallyConnHyPre}
  \begin{algorithmic}[1]
  \REQUIRE Equivalent fully-digital beamfocusing matrix $\mathbf{W}$, iteration index $ m_{\mathrm{P}}=0 $, convergence tolerance $ 0 < \varepsilon \ll 1 $, the number of maximum iterations $ N_{\mathrm{max}} $.
  \STATE Initialize the point $ \mathbf{T}_\mathrm{A}^{(0)}, \mathbf{T}_\mathrm{D}^{(0)}$. Calculate the initial objective function value $ \mathrm{Obj}^{(0)} $.
  \REPEAT
    \STATE Update $ m_{\mathrm{P}}=m_{\mathrm{P}}+1 $;
    \STATE Update $ \mathbf{T}_\mathrm{D}$ and $ \mathbf{T}_\mathrm{A} $, based on Eqs. \eqref{equPrePartDigSolution}, \eqref{partConAnaPreSolution}, respectively; 
    \STATE Calculate the objective value $ \mathrm{Obj}^{(m_{\mathrm{P}})} $;
  \UNTIL convergence.
  \end{algorithmic}
\end{algorithm} 
\begin{subequations}\label{equivalentPrecoderSolve} 
  \begin{align}
    \mathop{\mathrm{minimize}}\limits_{\mathbf{T}_\mathrm{A},\mathbf{T}_\mathrm{D}}\quad   &\|\mathbf{W}-\mathbf{T}_\mathrm{A}\mathbf{T}_\mathrm{D}\|_\mathrm{F}^{2} \tag{\ref{equivalentPrecoderSolve}}\\
    \mathrm{subject \;to} \quad  &\|\mathbf{T}_\mathrm{A}\mathbf{T}_\mathrm{D}\|_\mathrm{F}^2= P,\label{equivalentPrecoderSolve_PowCon}\\
    &\mathbf{T}_\mathrm{A}\in\mathcal{A}.\label{equivalentPrecoderSolve_UnitModuCon}
    \vspace{-1em}
   \end{align}
\end{subequations}
Due to the unit-modulus constraint on the analog beamformer, problem \eqref{equivalentPrecoderSolve} is nonconvex. As such, analog-and-digital beamformers are optimized iteratively via alternating optimization (AO), and the details are shown as follows.
\vspace{-0.4em}
\subsection{Fully-Connected Architecture}

\subsubsection{Update Digital Precoder $ \mathbf{T}_\mathrm{D}$}
With the given analog beamformer $ \mathbf{T}_\mathrm{A} $, we optimize the digital beamformer $ \mathbf{T}_\mathrm{D} $. 
The problem then becomes a well-known least squares problem \cite{AAroraHybrid}. The solution at the $ m_{\mathrm{F}} $-th iteration of AO is given by
\begin{equation}\label{digitalPrecoderSolution}
  \mathbf{T}_\mathrm{D}=(\mathbf{T}_\mathrm{A}^H\mathbf{T}_\mathrm{A})^{-1}\mathbf{T}_\mathrm{A}^H\mathbf{W}.
  \vspace{-0.3em}
\end{equation}

\noindent To satisfy the power constraint, we normalize $ \mathbf{T}_\mathrm{D} $ by a factor of $ \frac{\sqrt{P}}{\|\mathbf{T}_\mathrm{A}\mathbf{T}_\mathrm{D}\|_\mathrm{F}} $. 
\subsubsection{Update Analog Precoder $\mathbf{T}_\mathrm{A}$}
When the digital beamformer is fixed, we employ the majorization-minimization algorithm to address the analog beamformer optimization problem \cite{AAroraHybrid}, where the solution at the $ m_{\mathrm{F}} $-th iteration \mbox{is given by}
\vspace{-0.2em}
\begin{equation}\label{fullyConnAnaPre} 
  \mathbf{T}_\mathrm{A}^{(m_{\mathrm{F}})}=e^{-j\arg([\mathbf{F}^{(m_{\mathrm{F}}-1)}]^T)},
  \vspace{-0.2em}
\end{equation}
where $ \mathbf{F}^{(m_{\mathrm{F}}-1)}\!=\!\mathbf{T}_\mathrm{D}\mathbf{W}^H\!-\!(\mathbf{M}_\mathrm{D}\!-\!\lambda_\mathrm{max}(\mathbf{M}_\mathrm{D})\mathbf{I})[\mathbf{T}_\mathrm{A}^{(m_{\mathrm{F}}-1)}]^H $, $ \mathbf{M}_\mathrm{D}\!=\!\mathbf{T}_\mathrm{D}\mathbf{T}_\mathrm{D}^H $ and denote its maximum eigenvalue as $\lambda_\mathrm{max}(\mathbf{M}_\mathrm{D})$. 
\vspace{-0.2em}

\subsection{Partially-Connected Architecture}
Compared to the fully-connected architecture, each RF chain connects to $ I_{\mathrm{g}} $ antennas in partially-connected architecture. Due to the block matrix form of $ \mathbf{T}_\mathrm{A} $, the power constraint \eqref{equivalentPrecoderSolve_PowCon} is simplified as $ \|\mathbf{T}_\mathrm{A}\mathbf{T}_\mathrm{D}\|_\mathrm{F}^2\!=\! I_{\mathrm{g}}\|\mathbf{T}_\mathrm{D}\|_\mathrm{F}^2\!=\!P $.
\subsubsection{Update Digital Precoder $\mathbf{T}_\mathrm{D}$}
With the fixed $ \mathbf{T}_\mathrm{A} $, the original problem is converted into a projection problem with a closed-form solution at the $ m_{\mathrm{P}} $-th iteration of AO, given by
\vspace{-0.5em}
\begin{equation}\label{equPrePartDigSolution}
    \mathbf{T}_\mathrm{D}=\sqrt{\frac{P}{I_\mathrm{g}}}\frac{\mathbf{T}_\mathrm{A}^H\mathbf{W}}{\|\mathbf{T}_\mathrm{A}^H\mathbf{W}\|_\mathrm{F}}.
\vspace{-0.2em}
\end{equation}

\subsubsection{Update Analog Precoder $\mathbf{T}_\mathrm{A}$}
With the given $ \mathbf{T}_\mathrm{D} $, the analog beamformer optimization problem is transformed into a row-wise form for each non-zero entry of $ \mathbf{T}_\mathrm{A} $ based on its block matrix structure. At the $ m_{\mathrm{P}} $-th iteration of AO, the problem is reformulated as follows
\vspace{-0.4em}
\begin{subequations}\label{PartConnAnaPre} 
  \begin{align}
    \mathop{\mathrm{minimize}}\limits_{\theta_{p,q}}\quad   &\sum_{p}\sum_{q}\|\mathbf{w}_{p}-e^{j\theta_{p,q}}\mathbf{d}_{q}\|_{2}^{2} \tag{\ref{PartConnAnaPre}}\\
    \mathrm{subject \;to} \quad   
    &\theta_{p,q}\in(0,2\pi], \forall p, \forall q=\left\lceil\frac{p}{I_{\mathrm{g}}}\right\rceil,
   \end{align}
\end{subequations}
where $ \mathbf{w}_p=[\mathbf{W}_{p,:}]^T $, $ \mathbf{d}_q=[(\mathbf{T}_\mathrm{D})_{q,:}]^T $, $(\cdot)_{p,:} $ represents the $ p $-th row for matrix, and $ \theta_{p,q} $ is the phase of the $(p,q)$ non-zero element of $ \mathbf{T}_\mathrm{A} $. Problem \eqref{PartConnAnaPre} is essentially a phase rotation problem. Then, the updated analog beamformer at the $ m_{\mathrm{P}} $-th iteration of AO is shown as
\vspace{-0.2em}
\begin{equation}\label{partConAnaPreSolution}
  [\mathbf{T}_\mathrm{A}]_{p,q}=e^{j\arg(\mathbf{d}_q^H\mathbf{w}_p)}, \forall p, \forall q=\left\lceil\frac{p}{I_{\mathrm{g}}}\right\rceil.  
\end{equation}
The procedures of fully-connected architecture algorithm are similar to that of partially-connected architecture. Here, only the procedures of partially-connected architecture are shown and summarized in \textbf{Algorithm \ref{partiallyConnHyPre}}.

\subsection{Computational Complexity and Convergence Analysis}
For the fully-connected architecture, the computational complexity of each iteration mainly depends on the pseudo-inverse operation for optimizing $\mathbf{T}_\mathrm{D}$, i.e., $ \mathcal{O}(N_\mathrm{t}^3) $. For the partially-connected architecture, the complexity is determined by Eqs. \eqref{equPrePartDigSolution} and \eqref{partConAnaPreSolution}, which are given by $\mathcal{O}(N_\mathrm{t}^\mathrm{RF}N_\mathrm{t}K)$ and $ \mathcal{O}(N_\mathrm{t}KN_\mathrm{t}^\mathrm{RF}) $, respectively. Let $ M_{\mathrm{F}} $ and $  M_{\mathrm{P}}  $ denote the average iteration numbers for the fully and partially-connected architecture, respectively. Then, the overall computational complexity for fully-connected and partially-connected architecture is given by $\mathcal{O}(M_{\mathrm{F}}N_\mathrm{t}^3)$ and $ \mathcal{O}(M_{\mathrm{P}}N_\mathrm{t}^\mathrm{RF}N_\mathrm{t}K) $, respectively. Both algorithms employ the AO algorithm, ensuring ultimate convergence \cite{ABeckOn}.

\section{Numerical Results}
\begin{figure}[t]
  \centering
  \epsfxsize=7.0in
  \includegraphics[scale=0.5]{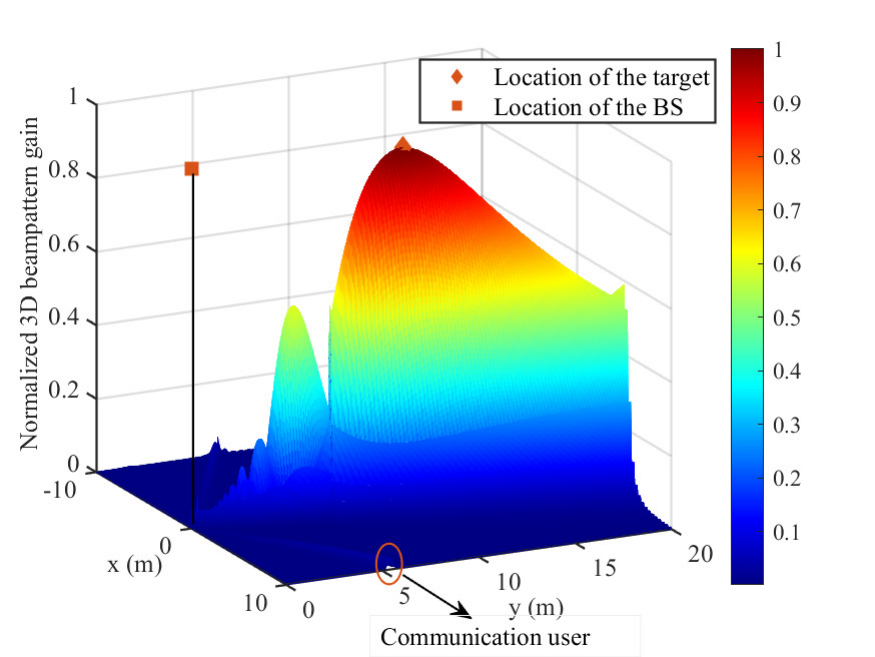}
  \caption{Near-field beamfocusing.}
  \label{figBeamfocusing}
  \vspace{-0.4em}
\end{figure}
\vspace{0.3em}
\begin{figure*}[ht]
  \vspace{-0.6em}
	\captionsetup[subfigure]{justification=centering} 
  \centering
	\subfloat[Distance estimation.]{
    \epsfxsize=7.0in
    \includegraphics[scale=0.5]{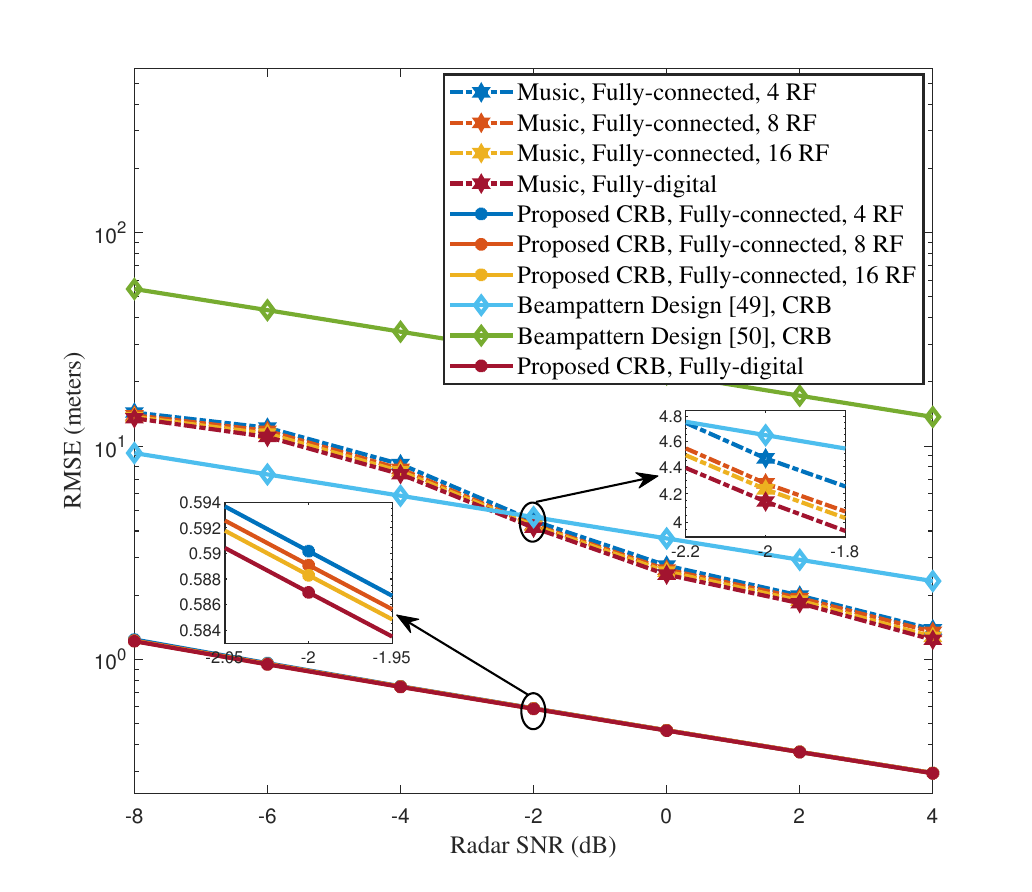}
		\label{figPowCRBFCDis}}
	\hfil
	\subfloat[Angle estimation.]{
    \epsfxsize=7.0in
    \includegraphics[scale=0.5]{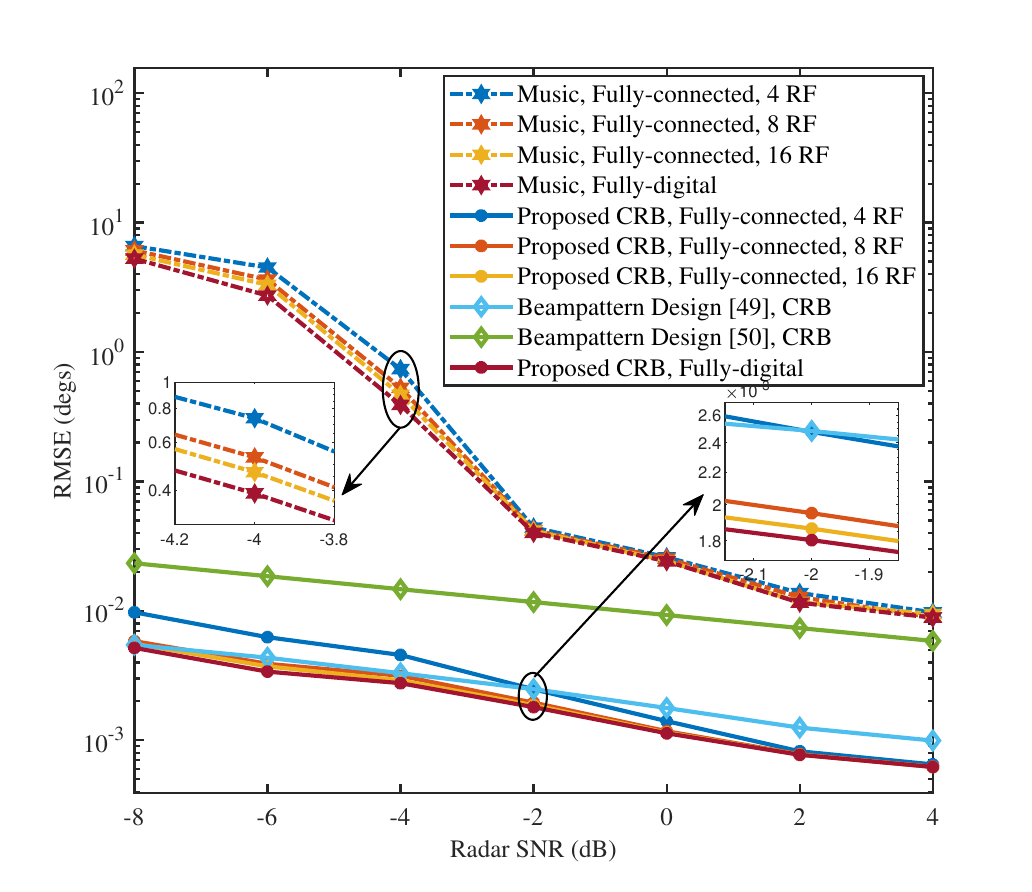}
		\label{figPowCRBFCAngle}}
	\caption{Parameters estimation performance in a point target scenario for fully-connected architecture at different SNRs and number of RF chains.}
	\label{figPowCRBFC}
  \vspace{-1em}
\end{figure*} 
In this section, numerical results are presented to confirm the performance of the proposed near-field hybrid ISAC framework. We assume that the BS is equipped with $N_\mathrm{t}\!=\!N_\mathrm{r}\!=\!64$ antennas, operating at the carrier frequency $ f\!=\!28\,\mathrm{GHz}\,(\lambda=0.0107\,\mathrm{m})$. The Rayleigh distance is calculated as $ \frac{2(N_\mathrm{t}\lambda/2)^2}{\lambda}\!=\!21.92\,\mathrm{m} $. The maximum transmit power and system static circuit power at the BS are set to $ P_{\mathrm{max}}\!=\!34\,\,\mathrm{dBm}, P_0\!=\!15\,\mathrm{dBm}$, respectively. The threshold of system EE is set to $\eta_\mathrm{th}\!=\!4\,\,\mathrm{bps/Hz/J}$. The power amplifier coefficient and the frame length are set to $\rho\!=\!0.5 $ and $ L\!=\!64 $ \cite{ZHeEnergy}. The noise variance of the BS and users are set to $ \sigma_{n}^2\!=\!\sigma_{k}^2\!=\!0\,\,\mathrm{dBm}$. Additionally, we assume that four users are served by the ISAC BS in the near-field region, located at distances of $ [15,10,15,10] $ meters and angles of $[-60^{\circ}, -30^{\circ}, 30^{\circ}, 60^{\circ}]$. The values of $ \beta_L $ and $ \beta_i $ for the $k$-th user can be calculated by adopting the method in \cite{CLinSubarrayBased}. Moreover, for a point target scenario, the target of interest is assumed to be located at the direction of $ 15^{\circ} $ and the distance $ 10\,\mathrm{m} $. For an extended target scenario, the TRM, $ \mathbf{B} $, follows the i.i.d. CSCG distribution with zero mean and variance $ \sigma_\mathrm{\bm\beta}^2=0\,\mathrm{dBm} $. 


\vspace{-0.6em}
\subsection{Feasibility of The Proposed Beamfocusing Scheme}
To better demonstrate the feasibility of the proposed CRB beamfocusing scheme, the number of transmit and receive antennas is set to 120 exclusively for this visualization purpose. The resulting three-dimensional (3D) beamfocusing gain diagram is illustrated in Fig. \ref{figBeamfocusing}. We set the BS at the location $ x=0 $ of x-axis and define the grid region with $ x\in\left[-20,20\right] $ $ \mathrm{m} $, $ y\in\left[0,40\right] $ $ \mathrm{m} $. Then, in Fig. \ref{figBeamfocusing}, the color gradient from blue to red represents an increasing gain, and the beamfocusing is performed at the location of the target while ensuring communication performance for users. It can be observed that the transmitted power is concentrated at the location with $ r=10.045 \,\, \mathrm{m} $ and $ \varphi=15.005\,^{\circ} $ that facilitates effective reflection. Therefore, under near-field scenarios, the proposed CRB beamfocusing scheme enables simultaneous estimation of the target's distance and angle.
\begin{figure*}[ht]
  \captionsetup[subfigure]{justification=centering} 
	\centering
	\subfloat[Distance estimation.]{
    \epsfxsize=7.0in
    \includegraphics[scale=0.5]{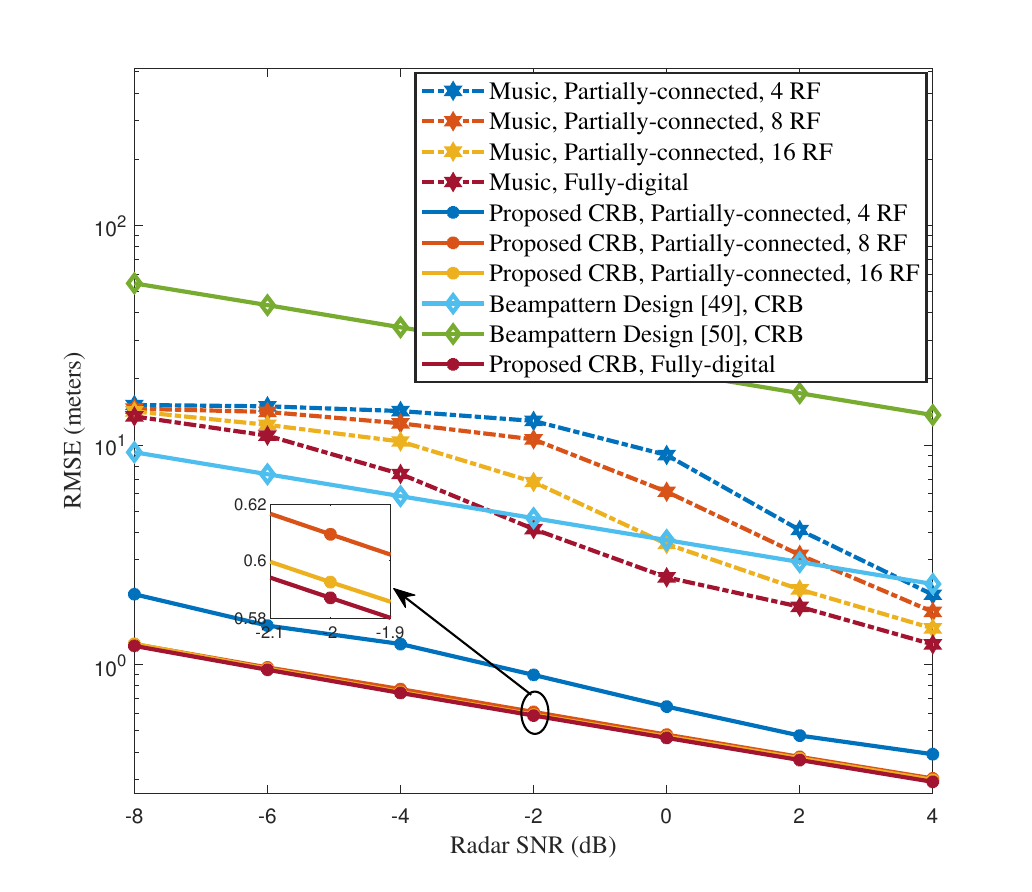}
		\label{figPowCRBPCDis}}
	\hfil
	\subfloat[Angle estimation.]{
    \epsfxsize=7.0in
    \includegraphics[scale=0.5]{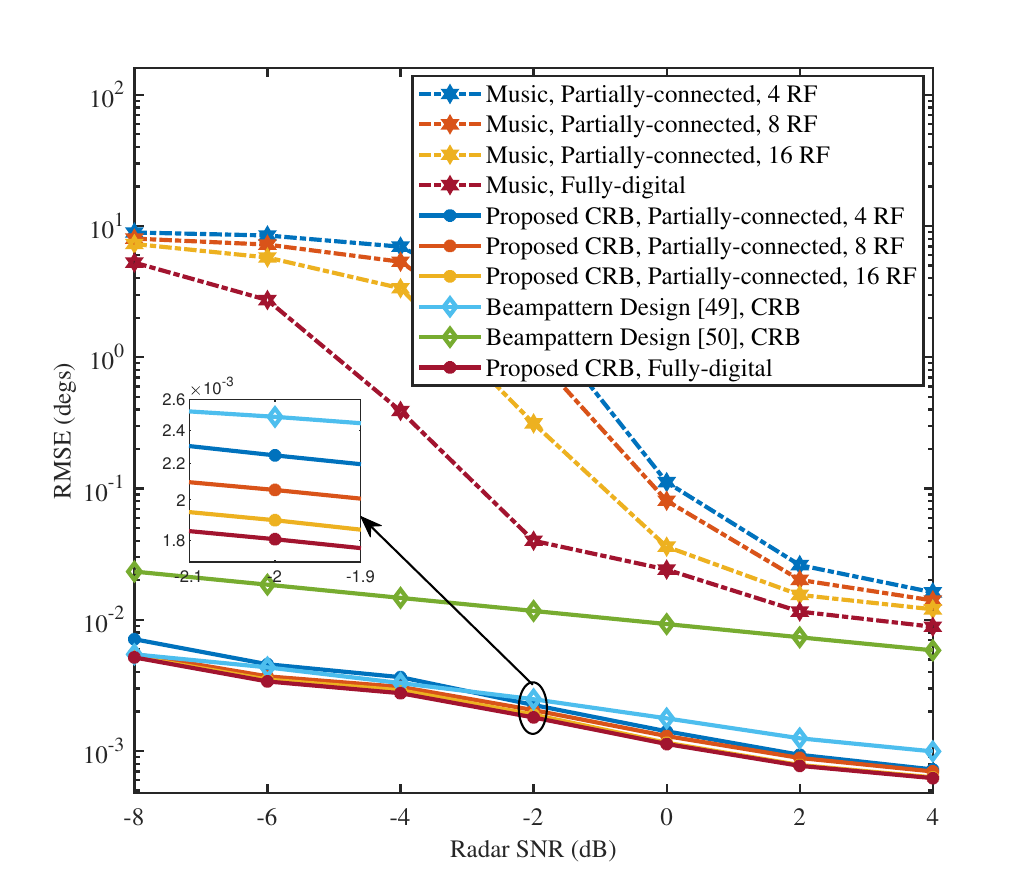}
		\label{figPowCRBPCAngle}}
	\caption{Parameters estimation performance in a point target scenario for partially-connected architecture at different SNRs and number of RF chains.}
	\label{figPowCRBPC}
  \vspace{-0.6em}
\end{figure*}

\begin{figure*}[ht]
  \vspace{-0.4em}
  \captionsetup[subfigure]{justification=centering} 
	\centering
	\subfloat[Distance estimation.]{
    \epsfxsize=7.0in
    \includegraphics[scale=0.5]{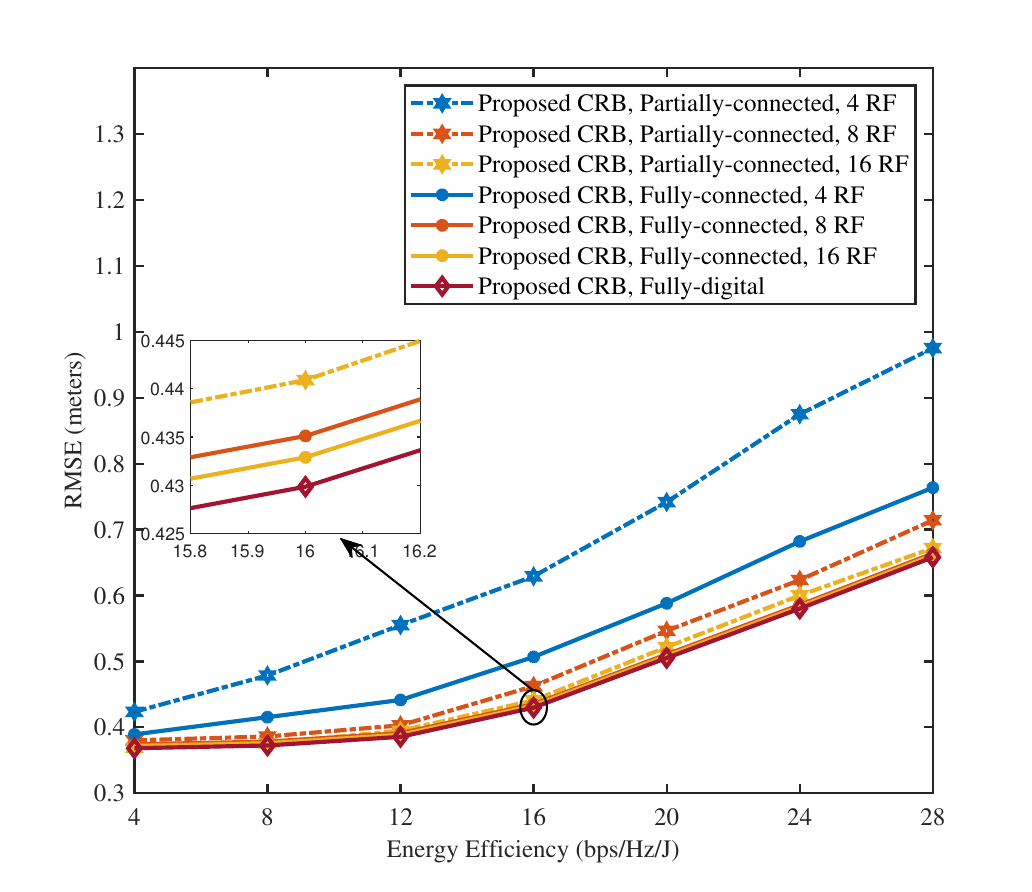}
		\label{figEECRBDis}}
	\hfil
	\subfloat[Angle estimation.]{
    \epsfxsize=7.0in
    \includegraphics[scale=0.5]{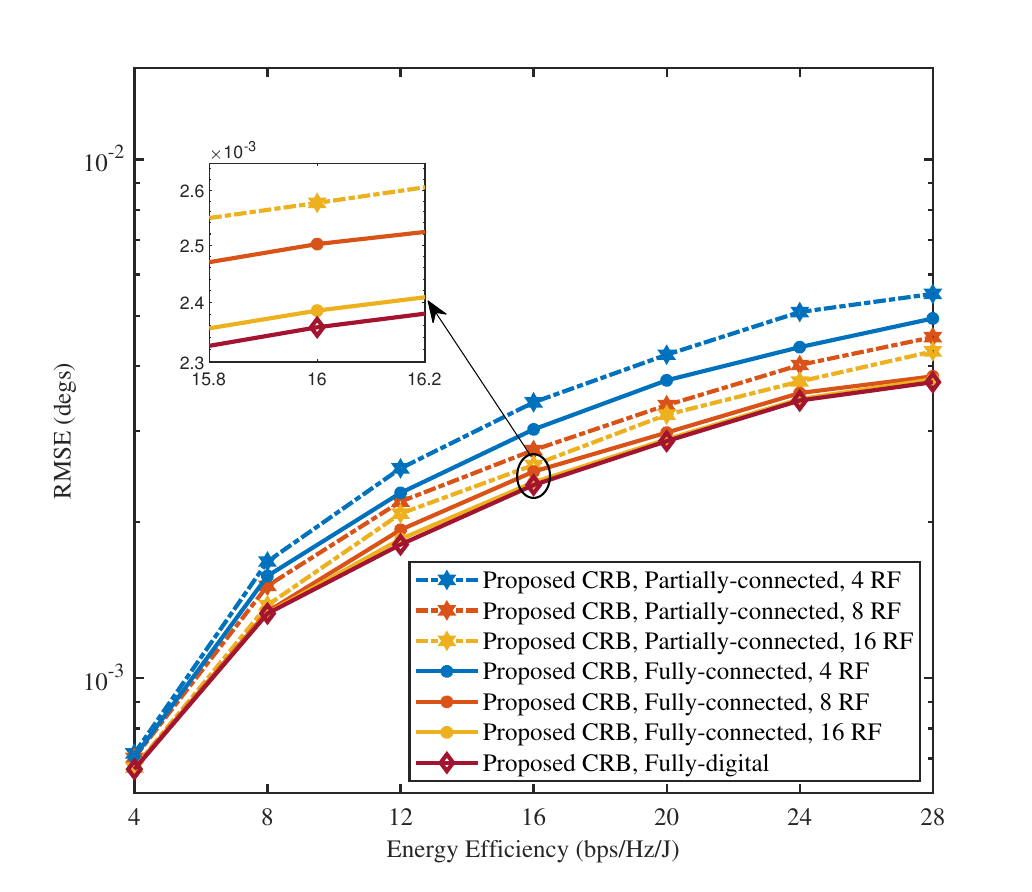}
		\label{figEECRBAngle}}
	\caption{Tradeoff between sensing and communication performance in the scenario of a point target.}
	\label{figEECRB}
  \vspace{-1em}
\end{figure*}
\vspace{-0.2em}
\subsection{Joint Beamfocusing for Point Target}
We evaluate the effectiveness of the proposed beamfocusing algorithm in a point target scenario. Specifically, the proposed CRB-based scheme is compared with several benchmark algorithms, including $\mathrm{2}\mathrm{D}$ MUSIC, and the methods proposed in \cite{GalappaththigeNear} and \cite{SunPingNear}. Furthermore, the parameter estimation performance of the proposed approach and the benchmark algorithms under both fully and partially-connected architectures is illustrated in Fig.\ref{figPowCRBFC} and Fig.\ref{figPowCRBPC}, respectively. With the radar signal-to-noise ratio (SNR) of the received echo signal gradually increasing, defined as $\frac{|\mu|^2LP}{\sigma_{n}^2}$ \cite{FLiuCramerRao}, the estimation performance of $\mathrm{2}\mathrm{D}$ MUSIC algorithm, algorithm proposed in \cite{GalappaththigeNear} and \cite{SunPingNear}, and CRB algorithm all improve accordingly. This is primarily because increasing the transmitted power enhances the performance of the estimation. Moreover, the distance estimation performance of the proposed CRB algorithm is superior to that of the $\mathrm{2}\mathrm{D}$ MUSIC algorithm as well as algorithm proposed in \cite{GalappaththigeNear} and \cite{SunPingNear}, thanks to the lower bound of the unbiased estimator. For the angle estimation of target, our method performs slightly worse than the approach in \cite{GalappaththigeNear} only under low-SNR conditions when the number of RF chains is 4. Nevertheless, in terms of overall angle estimation performance, our method still outperforms both baselines. Therefore, the proposed CRB algorithm is capable of achieving superior estimation performance for ISAC systems.\par
\begin{figure}[t]
  \captionsetup[subfigure]{justification=centering} 
  \centering
  \epsfxsize=7.0in
  \includegraphics[scale=0.5]{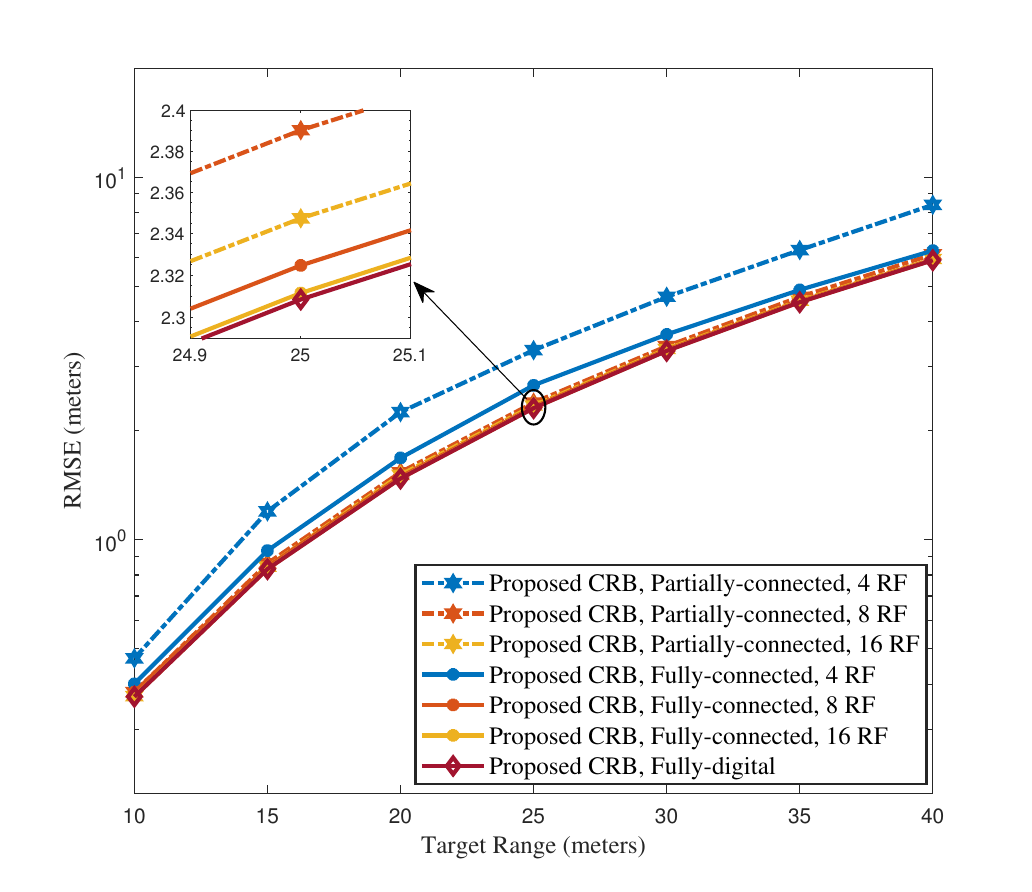}
  \caption{The CRB for different target distances.}
  \label{figCRBRanDis}
  \vspace{-1em}
\end{figure}
\begin{figure}[t]
  \centering
  \epsfxsize=7.0in
  \includegraphics[scale=0.5]{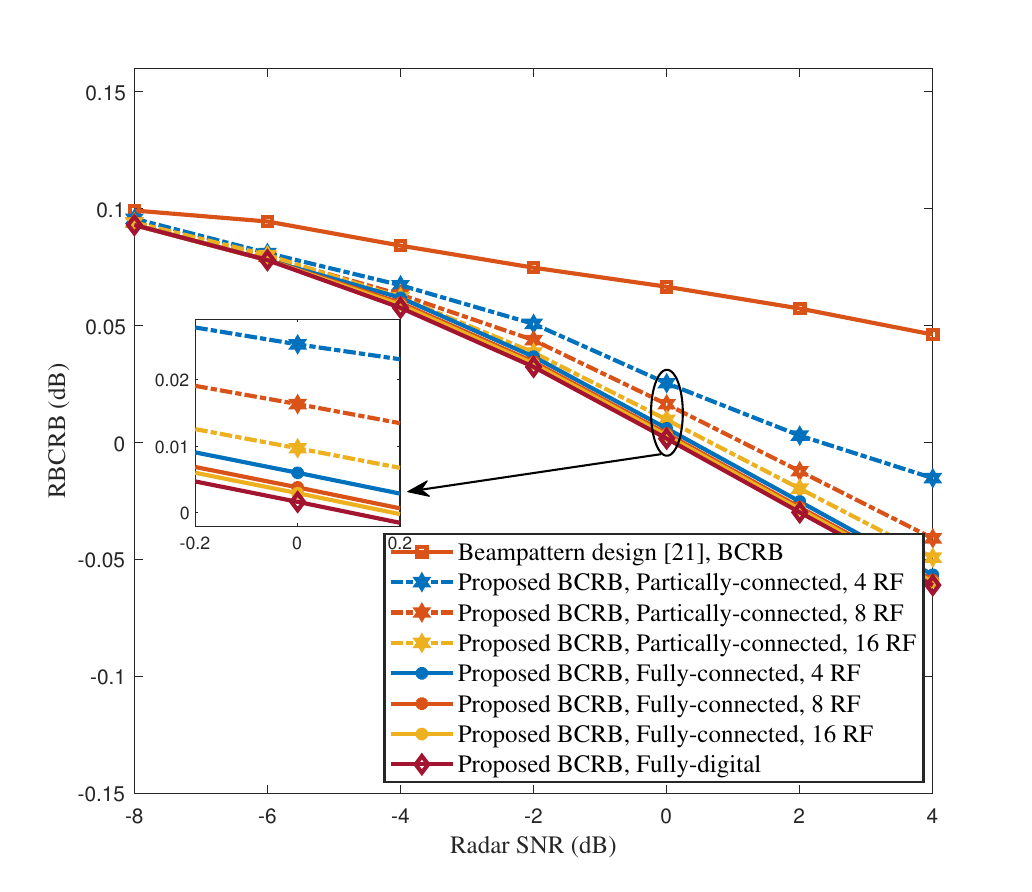}
  \caption{Parameters estimation performance in the scenario of an extended target at different SNRs and number of RF chains.}
  \label{figPowBCRB}
  \vspace{-0.6em}
\end{figure}
Fig. \ref{figPowCRBFC} and Fig. \ref{figPowCRBPC} indicate that the CRB for angle-and-distance estimation in fully and partially-connected architectures progressively approach the performance of the equivalent fully-digital solution, as the SNR increases. Compared to partially-connected architecture, the fully-connected architecture exhibits less degradation in parameter estimation performance. This is attributed to the higher DoF inherent in fully-connected architecture. Moreover, it can be observed that compared to the far-field scenario where only angle estimation is feasible, joint angle-and-distance estimation can be simultaneously achieved in the near-field region. As the SNR increases, the accuracy of distance estimation improves to centimeter level, \mbox{enabling this method to satisfy the} requirements of most high-accuracy localization applications.\par
In Fig. \ref{figEECRB}, we consider the tradeoff between the EE of the system and estimation performance. We set the maximum transmit power of the base station at $28$ $\mathrm{dBm}$ and the SINR threshold for communication users to $2$ $\mathrm{dB}$. As the EE of the system increases, the CRB of distance and angle estimation becomes higher, leading to larger estimation errors. The main reason is that the BS is forced to allocate more power for communication to enhance the EE of the system, reducing the available power allocated for sensing. Meanwhile, angle estimation is more sensitive to improvements in EE, compared to distance information, primarily because of the higher proportion of angle information in the steering vector. As the power allocated to sensing decreases, more energy is lost in the angle domain. \par
Subsequently, in Fig. \ref{figCRBRanDis}, we investigate the impact of target distance on the estimation performance for near-field ISAC systems. As the target distance $r$ increases, the RMSE of the estimated distance also increases as expected, indicating a gradual decline in distance estimation accuracy. As the target distance approaches the Rayleigh distance, the estimation accuracy degrades to the meter level. More importantly, once beyond the Rayleigh distance, the CRB of the distance estimation increases rapidly. This can be attributed to two main factors: first, as the distance increases, the near-field spherical wave model progressively approaches the conventional plane wave model; second, the attenuation of the target's echo signal becomes more pronounced with increasing distance. Consequently, the accuracy of distance estimation deteriorates.
\vspace{-1.6em}
\subsection{Joint Beamfocusing for Extended Target}
We present the simulation for an extended target estimation performance under both partially and fully hybrid architectures in Fig. \ref{figPowBCRB}. Compared to the conventional sensing metric of beamforming gain in \cite{ZHeEnergy}, our proposed BCRB-based optimization framework achieves superior estimation performance, even under hybrid architectures. This demonstrates that the BCRB serves as an effective performance metric for sensing, as it can enhance the overall estimation accuracy of the system. Moreover, as the radar SNR increases, the BCRB of both fully and partially connected architectures decreases, with diminishing returns. This effect primarily arises from the inherent singularity and limited DoF in its hybrid beamfocusing matrix. Specifically, only a few non-zero singular values increase with SNR, leading to estimation results that are increasingly dominated by prior information. Therefore, based on the simulation results, it can be concluded that, under a reasonable tradeoff between system cost and performance, appropriately increasing the number of RF chains in hybrid architectures can effectively reduce the BCRB, enhance the estimation performance of practical systems, and alleviate the reliance on prior information.

\section{Conclusion}

In this paper, we proposed an ISAC hybrid architecture beamfocusing design for the scenarios of a point and an extended target, respectively. Specifically, for the point target, we formulated an optimization problem to minimize the joint CRB of angle and distance for the point target, subject to constraints on the minimum required system EE and the SINR of communication users. For the extended target, the rank deficiency caused by hybrid beamfocusing renders the CRB of TRM estimation nonexistent. As an alternative, we utilized the BCRB as the optimization objective. Subsequently, we transformed the nonconvex optimization problems into the design of an equivalent fully-digital beamformer. Based on the equivalent fully-digital beamformer, we minimized the Euclidean distance between the product of the hybrid analog-and-digital beamformer and the equivalent fully-digital beamformer. Numerical results demonstrated that it is feasible to simultaneously estimate the distance and angle for the point target in the near field region. Additionally, a nontrivial tradeoff exists between the EE of the system and estimation performance, as higher EE of the system always leads to degraded estimation performance. Therefore, achieving a balance between the EE of the system and estimation performance is crucial in practical ISAC systems employing M-MIMO technology.


\begin{appendices}

\section{Derivation of FIM for Point Target}\label{FIMPointDeriveAppendix}
The partial derivatives of the received signal in \eqref{vecReceSignal} with respect to the vector of unknown parameters $ \bm\phi $ and $ \tilde{\bm\mu} $ are expressed as
\vspace{-0.2em}
\begin{align}
    &\frac{\partial\bm\upsilon}{\partial\bm\phi}=[\mu\mathrm{vec}(\dot{\mathbf{B}_{r}}\mathbf{X}),\,\mu\mathrm{vec}(\dot{\mathbf{B}_{\varphi} }\mathbf{X})] \label{pointPartialDerivative_r_phi},\\
    &\frac{\partial\bm\upsilon}{\partial\tilde{\bm\mu}}=[1,j]\otimes\mathrm{vec}(\mathbf{BX}), \label{pointPartialDerivative_mu}
    \vspace{-0.4em}
\end{align}
where $\dot{\mathbf{B}_{r}}$ and $\mathbf{B}_{\varphi}$ denote the partial derivatives of the TRM with respect to distance and angle, respectively. Subsequently, according to the definition of FIM \eqref{FIMDefinition} and expression \eqref{pointPartialDerivative_r_phi}, the elements of FIM with respect to the unknown parameters vector $ \bm\phi $ are given by
\vspace{-0.4em}
\begin{align}
  J_{uv}&=\frac{2}{\sigma_{n}^2}\mathrm{Re}\left\{\mu^*\mathrm{vec}(\dot{\mathbf{B}}_u\mathbf X)^H\mu\mathrm{vec}(\dot {\mathbf{B}}_v\mathbf X)\right\}\notag\\
  &=\frac{2|\mu|^2L}{\sigma_{n}^2}\mathrm{Re}\left\{\mathrm{Tr}\left(\dot{\mathbf{B}}_v\mathbf{R}_\mathrm{\mathbf{X}}\dot{\mathbf{B}}_u^H\right)\right\},
\end{align}
where the element $ J_{uv}, \forall u,v\in{\bm\phi} $. We then denote these elements as a block matrix $ \mathbf{J}_{\bm\phi \bm\phi} $, expressed as
\begin{equation}
  \mathbf{J}_{\bm\phi \bm\phi}=
  \begin{bmatrix} 
    J_{rr} & J_{r\varphi }\\
    J_{r\varphi}^T & J_{\varphi\varphi}
  \end{bmatrix}.
\end{equation}
In what follows, based on \eqref{pointPartialDerivative_mu}, \eqref{pointPartialDerivative_r_phi}, and \eqref{FIMDefinition}, the remaining elements of FIM can be divided into two block matrices, i.e., $ \mathbf{J}_{\bm\phi \tilde{\bm\mu}} $, $ \mathbf{J}_{\tilde{\bm\mu} \tilde{\bm\mu}} $, the specific expressions are given by
  \begin{align}
    \mathbf{J}_{\bm\phi \tilde{\bm\mu}}&=\frac{2}{\sigma_{n}^2}\mathrm{Re}\left\{
    \begin{bmatrix}
      \mu^*\mathrm{vec}(\dot{\mathbf{B}}_{r}\mathbf{X})^H\\
      \mu^*\mathrm{vec}(\dot{\mathbf{B}}_{\varphi}\mathbf X)^H
    \end{bmatrix}  
   \left([1,j]\otimes\mathrm{vec}(\mathbf{BX})\right) \right\}  \notag\\
   &=\frac{2L}{\sigma_{n}^2}\mathrm{Re}\left\{
    \begin{bmatrix}
      \mu^*\mathrm{Tr}(\dot{\mathbf{B}}_{r}^H\mathbf{B}\mathbf{R}_\mathrm{\mathbf{X}})\\
      \mu^*\mathrm{Tr}(\dot{\mathbf{B}}_{\varphi}^H\mathbf{B}\mathbf{R}_\mathrm{\mathbf{X}})
    \end{bmatrix}  
   [1,j] \right\},  \\
    \mathbf{J}_{\tilde{\bm\mu} \tilde{\bm\mu}}&=\frac{2}{\sigma_{n}^2}\mathrm{Re}\left\{
      \left([1,j]\otimes\mathrm{vec}(\mathbf{BX})\right)^H\left([1,j]\otimes\mathrm{vec}(\mathbf{BX})\right)\right\}\notag\\
    &=\frac{2L}{\sigma_{n}^2}\mathrm{Re}\left\{\left([1,j]^H[1,j]\right)
    \mathrm{Tr}\left(\mathbf{BTT}^H\mathbf{B}^H\right)\right\} \notag\\
    &=\frac{2L}{\sigma_{n}^2}\mathbf{I}_2\mathrm{Tr}\left(\mathbf{B}^H\mathbf{B}\mathbf{R}_\mathrm{\mathbf{X}}\right),
  \end{align}
respectively. This completes the proof.
\section{The proof of proposition 1} \label{propositionProof}
  Let $ f(\mathbf{X}) \!\geq \!f(\mathbf{Y}) $ for matrix variables $ \mathbf{X}\!\succeq \mathbf{Y}\!\succeq \!\mathbf{0} $, where $ f(\cdot) $ is a monotonically decreasing matrix function. Correspondingly, the function $ \mathrm{Tr}(\mathbf{X}^{-1}) $ is a matrix monotonically decreasing function with respect to the matrix variable $ \mathbf{X} $. On the basis of the theorem, we introduce the auxiliary matrix $ \bm\Xi \succeq \mathbf{0} $ to construct constraint $ \mathbf{J}_{\bm\phi\bm\phi}\!-\!\mathbf{J}_{\bm\phi\tilde{\bm\mu}}\mathbf{J}_{\tilde{\bm\mu}\tilde{\bm\mu}}^{-1}\mathbf{J}_{\bm\phi\tilde{\bm\mu}}^T\! \succeq \!\bm\Xi $. Therefore, minimize $ \mathrm{Tr}\{(\mathbf{J}_{\bm\phi\bm\phi}\!-\!\mathbf{J}_{\bm\phi\tilde{\bm\mu}}\mathbf{J}_{\tilde{\bm\mu}\tilde{\bm\mu}}^{-1}\mathbf{J}_{\bm\phi\tilde{\bm\mu}}^T)^{-1}\} $ is equivalent to minimize $ \mathrm{Tr}(\bm\Xi^{-1}) $. Then, we leverage the Schur complement to transform the constraint into a simple form \eqref{CRBconstraint_shurCon}. 
\section{Derivation of FIM for Extended Target}\label{FIMExtendedDeriveAppendix}
Define $ \bm\delta\!=\![\bm\beta_\mathrm{R}^T,\,\,\bm\beta_\mathrm{I}^T]^T \in \mathbb{R}^{2N_\mathrm{t}N_\mathrm{r}}$ as the unknown parameters vector, where $ \bm\beta\!=\!\mathrm{vec}(\mathbf{B})\in\mathbb{C}^{N_\mathrm{t}N_\mathrm{r}} $, $\bm\beta_\mathrm{R} $ and $ \bm\beta_\mathrm{I} $ are the real and imaginary parts of vector $ \bm\beta $, respectively. By vectorizing the received echo signal matrix for the extended target, the expression is rewritten as 
\vspace{-0.2em}  
\begin{equation}\label{vecRecSignalResMar}
  \mathbf{y}=\mathrm{vec}(\mathbf{Y})= (\mathbf{X}^T\otimes \mathbf{I}_{N_\mathrm{r}})\bm\beta+\mathbf{n}.
  \vspace{-0.2em}  
\end{equation}
Without loss of generality, we denote $ \bm\upsilon = (\mathbf{X}^T\otimes \mathbf{I}_{N_\mathrm{r}})\bm\beta $. The partial derivatives of the received signal in \eqref{vecRecSignalResMar} with respect to the unknown parameters vector $ \bm\beta_\mathrm{R} $ and $ \bm\beta_\mathrm{I} $ are given by
\vspace{-0.2em}
\begin{equation}
  \frac{\partial\bm\upsilon}{\partial\bm\beta_{\mathrm R}}=\mathbf{X}^T\otimes \mathbf{I}_{N_\mathrm{r}} \label{extendPartialDerivativeR},
  \vspace{-0.2em}
\end{equation}
\begin{equation}
  \frac{\partial\bm\upsilon}{\partial\bm\beta_{\mathrm I}}=j\mathbf{X}^T\otimes \mathbf{I}_{N_\mathrm{r}} \label{extendPartialDerivativeI}.
  \vspace{-0.2em}
\end{equation}
In light of \eqref{extendPartialDerivativeR}, \eqref{extendPartialDerivativeI}, and \eqref{FIMDefinition}, the FIM $ \mathbf{J}_{\bm\delta}\in\mathbb{R}^{2N_\mathrm{t}N_\mathrm{r}\times2N_\mathrm{t}N_\mathrm{r}} $ of the unknown parameters vector $ \bm\beta_\mathrm{R} $ and $ \bm\beta_\mathrm{I} $ can be obtained, and the FIM is partitioned into four block matrices, i.e., $\mathbf{J}_{\bm\beta_\mathrm{R}\bm\beta_\mathrm{R}},\,\mathbf{J}_{\bm\beta_\mathrm{R}\bm\beta_\mathrm{I}},\,\mathbf{J}_{\bm\beta_\mathrm{I}\bm\beta_\mathrm{R}},\,\mathbf{J}_{\bm\beta_\mathrm{I}\bm\beta_\mathrm{I}}$, with the specific expressions given by
\begin{equation}\label{FIMDiviDedMatrixExtendedCRLB}
   \mathbf{J}_{\bm\delta}=
\begin{bmatrix} 
   \mathbf{J}_{\bm\beta_\mathrm{R}\bm\beta_\mathrm{R}} & \mathbf{J}_{\bm\beta_\mathrm{R}\bm\beta_\mathrm{I}}\\
   \mathbf{J}_{\bm\beta_\mathrm{I}\bm\beta_\mathrm{R}} & \mathbf{J}_{\bm\beta_\mathrm{I}\bm\beta_\mathrm{I}}
\end{bmatrix},
\end{equation}
where the expressions of the block matrices $ \mathbf{J}_{\bm\beta_\mathrm{R}\bm\beta_\mathrm{R}} $, $ \mathbf{J}_{\bm\beta_\mathrm{R}\bm\beta_\mathrm{I}} $, $ \mathbf{J}_{\bm\beta_\mathrm{I}\bm\beta_\mathrm{R}} $, $ \mathbf{J}_{\bm\beta_\mathrm{I}\bm\beta_\mathrm{I}} $ are provided below 
\begin{align}
  \mathbf{J}_{\bm\beta_{\mathrm R}\bm\beta_{\mathrm R}}
  &=\frac{2}{\sigma_{n}^2}\mathrm{Re}\left\{(\mathbf{X}^T\otimes \mathbf{I}_{N_\mathrm{r}})^H(\mathbf{X}^T\otimes \mathbf{I}_{N_\mathrm{r}})\right\}\notag\\
  &=\frac{2L}{\sigma_{n}^2}\mathrm{Re}\left\{\mathbf{R}_\mathrm{\mathbf{X}}^T\otimes\mathbf{I}_{N_\mathrm{r}}\right\},\\
  \mathbf{J}_{\bm\beta_{\mathrm R}\bm\beta_{\mathrm I}}
  &=\frac{2}{\sigma_{n}^2}\mathrm{Re}\left\{(\mathbf{X}^T\otimes \mathbf{I}_{N_\mathrm{r}})^H(j\mathbf{X}^T\otimes \mathbf{I}_{N_\mathrm{r}})\right\}\notag\\
  &=-\frac{2L}{\sigma_{n}^2}\mathrm{Im}\left\{\mathbf{R}_\mathrm{\mathbf{X}}^T\otimes\mathbf{I}_{N_\mathrm{r}}\right\},\\
  \mathbf{J}_{\bm\beta_{\mathrm I}\bm\beta_{\mathrm R}}
  &=\frac{2}{\sigma_{n}^2}\mathrm{Re}\left\{(j\mathbf{X}^T\otimes \mathbf{I}_{N_\mathrm{r}})^H(\mathbf{X}^T\otimes \mathbf{I}_{N_\mathrm{r}})\right\}\notag\\
  &=\frac{2L}{\sigma_{n}^2}\mathrm{Im}\left\{\mathbf{R}_\mathrm{\mathbf{X}}^T\otimes\mathbf{I}_{N_\mathrm{r}}\right\},\\
  \mathbf{J}_{\bm\beta_{\mathrm I}\bm\beta_{\mathrm I}}
  &=\frac{2}{\sigma_{n}^2}\mathrm{Re}\left\{(j\mathbf{X}^T\otimes \mathbf{I}_{N_\mathrm{r}})^H(j\mathbf{X}^T\otimes \mathbf{I}_{N_\mathrm{r}})\right\}\notag\\
  &=\frac{2L}{\sigma_{n}^2}\mathrm{Re}\left\{\mathbf{R}_\mathrm{\mathbf{X}}^T\otimes\mathbf{I}_{N_\mathrm{r}}\right\}.
\end{align}

\begin{lemma}\label{lemmaComplexInverse}
   Matrix $ \mathbf{J}_{\bm\delta} $ is invertible if and only if matrix $ \mathbf{X} $ is invertible. Denoting $ \mathbf{Y}\!=\!\mathbf{X}^{-1} $, we have
\end{lemma}
\begin{equation}
  \begin{bmatrix}
    \mathrm{Re}(\mathbf{X})&&-\mathrm{Im}(\mathbf{X})\\\mathrm{Im}(\mathbf{X})&&\mathrm{Re}(\mathbf{X})
  \end{bmatrix}^{-1}
  =
  \begin{bmatrix}
    \mathrm{Re}(\mathbf{Y})&&-\mathrm{Im}(\mathbf{Y})\\\mathrm{Im}(\mathbf{Y})&&\mathrm{Re}(\mathbf{Y})
  \end{bmatrix}.
\end{equation}
Through Lemma \ref{lemmaComplexInverse}, we discover that CRB of TRM estimation for the extended target does not exist under hybrid structure ISAC systems, as the matrix $ \mathbf{R}_\mathrm{\mathbf{X}} $ is a singular, due to the rank-deficiency of the transmit signal $\mathbf{X} \in \mathbb{C}^{N_\mathrm{t} \times L}$, i.e.,
\begin{align}\label{dataStreamRank}
   \mathrm{rank}\left(\mathbf{X}\right)&\leq\min\left\{\mathrm{rank}\left(\mathbf{T}_\mathrm{A}\right),\mathrm{rank}\left(\mathbf{T}_\mathrm{D}\right),\mathrm{rank}\left(\mathbf{S}\right)\right\}\notag\\
   &=K\leq N_\mathrm{t}^\mathrm{RF} \leq N_\mathrm{t}\leq L.
\end{align}
The maximum number of transmitted data streams is capped at $ N_\mathrm{t}^\mathrm{RF} $, which is the number of transmit RF chains. Thus, the DoFs are insufficient to recover the rank-$ N_\mathrm{t} $ matrix $ \mathbf{B} $. 
\vspace{0.2em}
\end{appendices}

\bibliographystyle{IEEEtran}
\bibliography{IEEEabrv,reference}

\begin{thebibliography}{10}
\providecommand{\url}[1]{#1}
\csname url@samestyle\endcsname
\providecommand{\newblock}{\relax}
\providecommand{\bibinfo}[2]{#2}
\providecommand{\BIBentrySTDinterwordspacing}{\spaceskip=0pt\relax}
\providecommand{\BIBentryALTinterwordstretchfactor}{4}
\providecommand{\BIBentryALTinterwordspacing}{\spaceskip=\fontdimen2\font plus
\BIBentryALTinterwordstretchfactor\fontdimen3\font minus \fontdimen4\font\relax}
\providecommand{\BIBforeignlanguage}[2]{{%
\expandafter\ifx\csname l@#1\endcsname\relax
\typeout{** WARNING: IEEEtran.bst: No hyphenation pattern has been}%
\typeout{** loaded for the language `#1'. Using the pattern for}%
\typeout{** the default language instead.}%
\else
\language=\csname l@#1\endcsname
\fi
#2}}
\providecommand{\BIBdecl}{\relax}
\BIBdecl

\bibitem{FLiuIntegrated}
F.~Liu \emph{et~al.}, ``Integrated sensing and communications: {Toward} dual-functional wireless networks for {6G} and beyond,'' \emph{{IEEE} J. Sel. Areas Commun.}, vol.~40, no.~6, pp. 1728--1767, Jun. 2022.

\bibitem{ZhenyaoHeUnlocking}
Z.~He \emph{et~al.}, ``Unlocking potentials of near-field propagation: {ELAA}-empowered integrated sensing and communication,'' \emph{{IEEE} Commun. Mag.}, vol.~62, no.~9, pp. 82--89, Sep. 2024.

\bibitem{shi2023intelligent}
W.~Shi, W.~Xu, X.~You, C.~Zhao, and K.~Wei, ``Intelligent reflection enabling technologies for integrated and green {I}nternet-of-{E}verything beyond {5G}: {C}ommunication, sensing, and security,'' \emph{IEEE Wireless Commun.}, vol.~30, no.~2, pp. 147--154, Apr. 2023.

\bibitem{JAndrewZhangEnabling}
J.~A. Zhang \emph{et~al.}, ``Enabling joint communication and radar sensing in mobile networks—{A} survey,'' \emph{{IEEE} Commun. Surveys Tuts.}, vol.~24, no.~1, pp. 306--345, 1st Quart, 2022.

\bibitem{WXuEdge}
W.~Xu \emph{et~al.}, ``Edge learning for {B5G} networks with distributed signal processing: {Semantic} communication, edge computing, and wireless sensing,'' \emph{{IEEE} J. Sel. Topics Signal Process.}, vol.~17, no.~1, pp. 9--39, Jan. 2023.

\bibitem{GEAFrankenDoppler}
G.~E.~A. Franken, H.~Nikookar, and P.~V. Genderen, ``Doppler tolerance of {OFDM} coded radar signals,'' in \emph{Proc. Eur. Radar Conf.}, Manchester, UK, Sep. 2006, pp. 108--111.

\bibitem{YLiuDesign}
Y.~Liu, G.~Liao, Z.~Yang, and J.~Xu, ``Design of integrated radar and communication system based on {MIMO-OFDM} waveform,'' \emph{{IEEE} J. Syst. Eng. Electron.}, vol.~28, no.~4, pp. 669--680, Aug. 2017.

\bibitem{THuangMAJoRCom}
T.~Huang, N.~Shlezinger, X.~Xu, Y.~Liu, and Y.~C. Eldar, ``{MAJoRCom}: {A} dual-function radar communication system using index modulation,'' \emph{{IEEE} Trans. Signal Process.}, vol.~68, pp. 3423--3438, May 2020.

\bibitem{FLiuMUMIMO}
F.~Liu, C.~Masouros, A.~Li, H.~Sun, and L.~Hanzo, ``{MU-MIMO} communications with {MIMO} radar: {From} coexistence to joint transmission,'' \emph{{IEEE} Trans. Wireless Commun.}, vol.~17, no.~4, pp. 2755--2770, Apr. 2018.

\bibitem{XLiuJoint}
X.~Liu \emph{et~al.}, ``Joint transmit beamforming for multiuser {MIMO} communications and {MIMO} radar,'' \emph{{IEEE} Trans. Signal Process.}, vol.~68, pp. 3929--3944, Jun. 2020.

\bibitem{BKChalisePerformance}
B.~K. Chalise, M.~G. Amin, and B.~Himed, ``Performance tradeoff in a unified passive radar and communications system,'' \emph{{IEEE} Signal Process. Lett.}, vol.~24, no.~9, pp. 1275--1279, Sep. 2017.

\bibitem{ZHeFullDuplex}
Z.~He \emph{et~al.}, ``Full-duplex communication for {ISAC: Joint} beamforming and power optimization,'' \emph{{IEEE} J. Sel. Areas Commun.}, vol.~41, no.~9, pp. 2920--2936, Sep. 2023.

\bibitem{LiuZiangJoint}
Z.~Liu, S.~Aditya, H.~Li, and B.~Clerckx, ``Joint transmit and receive beamforming design in full-duplex integrated sensing and communications,'' \emph{{IEEE} J. Sel. Areas Commun.}, vol.~41, no.~9, pp. 2907--2919, Sep. 2023.

\bibitem{FLiuCramerRao}
F.~Liu, Y.-F. Liu, A.~Li, C.~Masouros, and Y.~C. Eldar, ``{Cram\'{e}r-Rao} bound optimization for joint radar-communication beamforming,'' \emph{{IEEE} Trans. Signal Process.}, vol.~70, pp. 240--253, Dec. 2022.

\bibitem{WXuToward}
W.~Xu, Y.~Huang, W.~Wang, F.~Zhu, and X.~Ji, ``Toward ubiquitous and intelligent {6G} networks: {From} architecture to technology,'' \emph{Sci. China Inf. Sci}, vol.~66, no.~3, pp. 130\,300:1--2, Mar. 2023.

\bibitem{OEAyachSpatially}
O.~E. Ayach, S.~Rajagopal, S.~Abu-Surra, Z.~Pi, and R.~W. Heath, ``Spatially sparse precoding in millimeter wave {MIMO} systems,'' \emph{{IEEE} Trans. Commun.}, vol.~13, no.~3, pp. 1499--1513, Mar. 2014.

\bibitem{FLiuHybrid}
F.~Liu and C.~Masouros, ``Hybrid beamforming with sub-arrayed {MIMO} radar: {Enabling} joint sensing and communication at {mmWave} band,'' in \emph{Proc. IEEE Int. Conf. Acoust., Speech Signal Process. (ICASSP)}, Brighton, UK, May 2019, pp. 7770--7774.

\bibitem{XWangPartially}
X.~Wang, Z.~Fei, J.~A. Zhang, and J.~Xu, ``Partially-connected hybrid beamforming design for integrated sensing and communication systems,'' \emph{{IEEE} Trans. Commun.}, vol.~70, no.~10, pp. 6648--6660, Oct. 2022.

\bibitem{KNRSVPrasadEnergy}
K.~N. R. S.~V. Prasad, E.~Hossain, and V.~K. Bhargava, ``Energy efficiency in massive {MIMO}-based {5G} networks: {Opportunities} and challenges,'' \emph{{IEEE} Wireless Commun.}, vol.~24, no.~3, pp. 86--94, Jun. 2017.

\bibitem{YaoJiachengEnergy}
J.~Yao, W.~Xu, G.~Zhu, K.~Huang, and S.~Cui, ``Energy-efficient edge inference in integrated sensing, communication, and computation networks,'' \emph{{IEEE} J. Sel. Areas Commun.}, pp. 1--1, May 2025.

\bibitem{AlluRavitejaRobust}
R.~Allu, M.~Katwe, K.~Singh, T.~Q. Duong, and C.-P. Li, ``Robust energy efficient beamforming design for {ISAC} full-duplex communication systems,'' \emph{{IEEE} Wireless Commun. Lett.}, vol.~13, no.~9, pp. 2452--2456, Sep. 2024.

\bibitem{ZHeEnergy}
Z.~He, W.~Xu, H.~Shen, Y.~Huang, and H.~Xiao, ``Energy efficient beamforming optimization for integrated sensing and communication,'' \emph{{IEEE} Wireless Commun. Lett.}, vol.~11, no.~7, pp. 1374--1378, Jul. 2022.

\bibitem{LYouBeam}
L.~You \emph{et~al.}, ``Beam squint-aware integrated sensing and communications for hybrid massive {MIMO LEO} satellite systems,'' \emph{{IEEE} J. Sel. Areas Commun.}, vol.~40, no.~10, pp. 2994--3009, Oct. 2022.

\bibitem{ZouJiaqiEnergy}
J.~Zou \emph{et~al.}, ``Energy-efficient beamforming design for integrated sensing and communications systems,'' \emph{{IEEE} Trans. Commun.}, vol.~72, no.~6, pp. 3766--3782, Jun. 2024.

\bibitem{CuiMingyaoNearField}
M.~Cui, Z.~Wu, Y.~Lu, X.~Wei, and L.~Dai, ``Near-field {MIMO} communications for {6G}: {Fundamentals}, challenges, potentials, and future directions,'' \emph{{IEEE} Commun. Mag.}, vol.~61, no.~1, pp. 40--46, Jan. 2023.

\bibitem{YDHuangNearField}
Y.-D. Huang and M.~Barkat, ``{Near}-field multiple source localization by passive sensor array,'' \emph{{IEEE} Trans. Antennas Propag.}, vol.~39, no.~7, pp. 968--975, Jul. 1991.

\bibitem{ZWangNearField}
Z.~Wang, X.~Mu, and Y.~Liu, ``Near-field integrated sensing and communications,'' \emph{{IEEE} Commun. Lett.}, vol.~27, no.~8, pp. 2048--2052, Aug. 2023.

\bibitem{DGalappaththigeNear}
D.~Galappaththige, S.~Zargari, C.~Tellambura, and G.~Y. Li, ``Near-field {ISAC}: {Beamforming} for multi-target detection,'' \emph{{IEEE} Wireless Commun. Lett.}, vol.~13, no.~7, pp. 1938--1942, Jul. 2024.

\bibitem{WenhaoHuIntegrated}
W.~Hu, Z.~He, W.~Xu, J.~Wang, and D.~W.~K. Ng, ``Integrated sensing and energy-efficient communication with near-field beamfocusing,'' in \emph{Proc. 2024 IEEE Global Commun. Conf. Workshops (GC Wkshps)}, Cape Town, South Africa, Dec. 2024.

\bibitem{LuShihangAjoint}
S.~Lu, X.~Meng, Z.~Du, Y.~Xiong, and F.~Liu, ``A joint radar-communication precoding design based on {Cram\'{e}r-Rao} bound optimization,'' in \emph{Proc. IEEE Int. Conf. Commun.}, Rome, Italy, Jul. 2023, pp. 2735--2740.

\bibitem{shi2024secrecy}
W.~Shi \emph{et~al.}, ``On secrecy performance of {RIS}-assisted {MISO} systems over {R}ician channels with spatially random eavesdroppers,'' \emph{IEEE Trans. Wireless Commun.}, vol.~23, no.~8, pp. 8357--8371, Aug. 2024.

\bibitem{JZouEnergy}
J.~Zou, Y.~Cui, Y.~Liu, and S.~Sun, ``Energy efficiency optimization for integrated sensing and communications systems,'' in \emph{Proc. IEEE Wireless Commun. Netw. Conf. (WCNC)}, Austin, TX, USA, Apr. 2022, pp. 216--221.

\bibitem{MCuiNear}
M.~Cui and L.~Dai, ``Near-field wideband beamforming for extremely large antenna arrays,'' \emph{{IEEE} Trans. Wireless Commun.}, vol.~23, no.~10, pp. 13\,110--13\,124, Oct. 2024.

\bibitem{LinAdaptive}
C.~Lin and G.~Y. Li, ``Adaptive beamforming with resource allocation for distance-aware multi-user indoor terahertz communications,'' \emph{{IEEE} Trans. Commun.}, vol.~63, no.~8, pp. 2985--2995, Aug. 2015.

\bibitem{CLinSubarrayBased}
C.~Lin, G.~Y. Li, and L.~Wang, ``Subarray-based coordinated beamforming training for {mmWave} and sub-{THz} communications,'' \emph{{IEEE} J. Sel. Areas Commun.}, vol.~35, no.~9, pp. 2115--2126, Sep. 2017.

\bibitem{BTangSpectrally}
B.~Tang and J.~Li, ``Spectrally constrained {MIMO} radar waveform design based on mutual information,'' \emph{{IEEE} Trans. Signal Process.}, vol.~67, no.~3, pp. 821--834, Feb. 2019.

\bibitem{MBellInformation}
M.~R. Bell, ``Information theory and radar waveform design,'' \emph{{IEEE} Trans. Inf. Theory}, vol.~39, no.~5, pp. 1578--1597, Sep. 2019.

\bibitem{ALeshemInformation}
A.~Leshem, O.~Naparstek, and A.~Nehorai, ``Information theoretic adaptive radar waveform design for multiple extended targets,'' \emph{{IEEE} J. Sel. Topics Signal Process.}, vol.~1, no.~1, pp. 42--55, Jun. 2007.

\bibitem{SMKayFundamentals}
S.~M. Kay, \emph{{Fundamentals of Statistical Signal Processing: Estimation Theory}}.\hskip 1em plus 0.5em minus 0.4em\relax Englewood Cliffs, NJ, USA: Prentice-Hall, 1993.

\bibitem{JXuSecure}
J.~Xu, W.~Xu, D.~W.~K. Ng, and A.~L. Swindlehurst, ``Secure communication for spatially sparse millimeter-wave massive {MIMO} channels via hybrid precoding,'' \emph{{IEEE} Trans. Commun.}, vol.~68, no.~2, pp. 887--901, Feb. 2020.

\bibitem{CGTsinosJoint}
C.~G. Tsinos, A.~Arora, S.~Chatzinotas, and B.~Ottersten, ``Joint transmit waveform and receive filter design for dual-function radar-communication systems,'' \emph{{IEEE} J. Sel. Topics Signal Process.}, vol.~15, no.~6, pp. 1378--1392, Nov. 2021.

\bibitem{WXuDisentangled}
W.~Xu, J.~Wu, S.~Jin, X.~You, and Z.~Lu, ``Disentangled representation learning empowered {CSI} feedback using implicit channel reciprocity in {FDD} massive {MIMO},'' \emph{{IEEE} Trans. Wireless Commun.}, vol.~23, no.~10, pp. 15\,169--15\,184, Oct. 2024.

\bibitem{AAroraHybrid}
A.~Arora, C.~G. Tsinos, B.~S. M.~R. Rao, S.~Chatzinotas, and B.~Ottersten, ``Hybrid transceivers design for large-scale antenna arrays using majorization-minimization algorithms,'' \emph{{IEEE} Trans. Signal Process.}, vol.~68, pp. 701--714, Dec. 2020.

\bibitem{ZqLuoSemidefinite}
Z.-Q. Luo, W.-K. Ma, A.~M.-C. So, Y.~Ye, and S.~Zhang, ``Semidefinite relaxation of quadratic optimization problems,'' \emph{{IEEE} Signal Process. Mag.}, vol.~27, no.~3, pp. 20--34, May 2010.

\bibitem{XYuRobust}
X.~Yu, D.~Xu, Y.~Sun, D.~W.~K. Ng, and R.~Schober, ``Robust and secure wireless communications via intelligent reflecting surfaces,'' \emph{{IEEE} J. Sel. Areas Commun.}, vol.~38, no.~11, pp. 2637--2652, Nov. 2020.

\bibitem{KYWangOutage}
K.-Y. Wang, A.~M.-C. So, T.-H. Chang, W.-K. Ma, and C.-Y. Chi, ``Outage constrained robust transmit optimization for multiuser {MISO} downlinks: {Tractable} approximations by conic optimization,'' \emph{{IEEE} Trans. Signal Process.}, vol.~62, no.~21, pp. 5690--5705, Nov. 2014.

\bibitem{ABeckSequential}
A.~Beck, A.~Ben-Tal, and L.~Tetruashvili, ``Sequential parametric convex approximation method with applications to nonconvex truss topology design problems,'' \emph{J. Global Optim.}, vol.~47, no.~1, pp. 29--51, May 2010.

\bibitem{JTabrikianTheoretical}
J.~Tabrikian and J.~L. Krolik, ``Theoretical performance limits on tropospheric refractivity estimation using point-to-point microwave measurements,'' \emph{{IEEE} Trans. Antennas Propag.}, vol.~47, no.~11, pp. 1727--1734, Nov. 1999.

\bibitem{ZJWuExtendedTarget}
Z.-J. Wu, C.-X. Wang, Y.-C. Li, and Z.-Q. Zhou, ``Extended target estimation and recognition based on multimodel approach and waveform diversity for cognitive radar,'' \emph{{IEEE} Trans. Geosci. Remote Sens.}, vol.~60, no.~9, pp. 1--14, Mar. 2022, {Art}. no. 5101014.

\bibitem{ABeckOn}
A.~Beck, ``On the convergence of alternating minimization for convex programming with applications to iteratively reweighted least squares and decomposition schemes,'' \emph{SIAM J. Optim.}, vol.~25, no.~1, pp. 185--209, Jan. 2015.

\bibitem{GalappaththigeNear}
D.~Galappaththige, S.~Zargari, C.~Tellambura, and G.~Y. Li, ``Near-field isac: Beamforming for multi-target detection,'' \emph{{IEEE} Wireless Commun. Lett.}, vol.~13, no.~7, pp. 1938--1942, Jul. 2024.

\bibitem{SunPingNear}
P.~Sun and B.~Wang, ``Near-field beam focusing for integrated sensing and communication systems,'' \emph{{IEEE} Internet Things J.}, vol.~12, no.~12, pp. 18\,643--18\,650, Jun. 2025.

\end{thebibliography}
\vspace{11pt}


\end{sloppypar}
\end{document}